\documentclass[12pt]{article}

\usepackage{amsmath,amsthm,amsfonts,amscd,latexsym,amssymb,amsbsy,dsfont,mathrsfs,color}
\usepackage[small, centerlast, it]{caption}
\usepackage{epsfig}
\usepackage[titles]{tocloft}
\usepackage[latin1]{inputenc}
\usepackage{verbatim} %for comments via \begin{comment}
\usepackage[active]{srcltx} %inverse search
\usepackage{fancyhdr}
\usepackage{caption}
\usepackage{enumerate}
\usepackage{color}
\usepackage{hyperref}
\hypersetup{
    hidelinks
}
\definecolor{NiColor}{RGB}{77,77,255}
\definecolor{NiColoRed}{RGB}{255,77,77}
\definecolor{NiCitation}{RGB}{77,255,77}

\def\sp{\hskip -5pt}
\def\spa{\hskip -3pt}

\newsymbol\rest 1316    
\def\emptyset{\varnothing} 
\def\b1{{1\!\!1}}
\def\pperp{{\perp\!\!\!\!\perp}}

\def\cB{\mathscr{B}}

\def\cF{{\ca F}}

\def\cH{{\ca H}}

\def\cL{\mathscr{L}}

\def\sS{{\mathsf S}}

\def\sV{\mathsf{V}}
\def\sT{{\mathsf T}}

\def\bC{{\mathbb C}}           %%%  complex numbers and so on

\def\bM{{\mathbb M}}

\def\bR{{\mathbb R}}

\def\gA{{\mathfrak A}}       %%% Ghotic
\def\gB{{\mathfrak B}}

\def\gD{{\mathfrak D}}

\def\beq{\begin{eqnarray}}
\def\eeq{\end{eqnarray}}

\newcommand{\ca}[1]{{\cal #1}}         %%  calligraphic

\usepackage{sectsty}
\sectionfont{\normalsize}%\normalfont
\subsectionfont{\normalsize\normalfont\itshape}
\newtheoremstyle{TheoremStyle}% <name>
{3pt}% <Space above>
{3pt}% <Space below>
{\slshape}% <Body font>
{}% <Indent amount>
{\bf}%{\itshape}% <Theorem head font>
{:}% <Punctuation after theorem head>
{.5em}% <Space after theorem head>
{}% <Theorem head spec (can be left empty, meaning 'normal')>

\theoremstyle{TheoremStyle}
\newtheorem{theorem}{Theorem}[section]
\newtheorem{corollary}[theorem]{Corollary}
\newtheorem{proposition}[theorem]{Proposition}
\newtheorem{lemma}[theorem]{Lemma}
\newtheorem{definition}[theorem]{Definition}
\newtheorem{remark}[theorem]{Remark}
\begin{document}

%%%%%%%%%%%%%   Title %%%%%%%%%%%%%%%%%%%%%%%%%%

\hfill{\sl Revised Version --  August 2026} 
\par 
\bigskip 
\par 
\rm 
\hfill Dedicated to the memory and work of Paul Busch
\\

%%%%%%%%%%%%%   Title %%%%%%%%%%%%%%%%%%%%%%%%%%

\par
\bigskip
\large
\noindent
{\bf  Spatial Localization of Relativistic Quantum Systems: The Commutativity Requirement and the Locality Principle. \\ Part I: A General Analysis}
\bigskip
\par
\rm
\normalsize 

%%%%%%%%%%%%%%%%%%%%%%%%%%%%%%%%%%%%%%%%%%%%%
%%%%%%%%%%%% Authors %%%%%%%%%%%%%%%%%%%%%%%%

\noindent  {\bf Valter Moretti$^{a}$ }\\
\par

\noindent 
  Department of  Mathematics, University of Trento, and INFN-TIFPA \\
 via Sommarive 14, I-38123  Povo (Trento), Italy.\\
 $^a$valter.moretti@unitn.it, corresponding author\\

 \normalsize

\par

\rm\normalsize

%\linespread{1.5}
\rm\normalsize

%%%%%%%%%%%% Date %%%%%%%%%%%%%%%%%%%%%%%%%%

\par
%\bigskip

%\noindent
\small
%{\bf Abstract}.

\begin{abstract}
We study the role of commutativity in the representation of relativistic locality for localization observables of relativistic quantum systems in Minkowski spacetime. A celebrated no-go theorem by Halvorson and Clifton shows that the commutativity of localization effects associated with causally separated regions is incompatible with other apparently natural assumptions about spatial localization. In particular, commutativity is assumed to provide the mathematical representation of locality in the Araki-Haag-Kastler formulation of quantum field theory. This raises the question of whether commutativity follows from more elementary locality principles of quantum theory.
Adopting Busch's operational analysis in terms of no-signaling and relativistic consistency, we argue that, for a particle-like system, the commutativity requirement does not follow from these principles. Under a natural local detectability principle, elementary localization observables are not localized in arbitrarily small spacetime neighborhoods of the corresponding spatial regions, but rather in regions containing the entire rest space (a Cauchy surface) on which the measurement is performed. This is a consequence of the very nature of a particle, which is assumed to be localized at a unique position on a rest space completely filled with ideal detectors. In this respect, there is no direct conflict with the Araki-Haag-Kastler formulation of local quantum physics.
On the other hand, we also show that commutativity and localization may coexist when one refers to less idealized localization procedures. Indeed, we introduce conditional localization POVMs associated with bounded spatial regions interpreted as laboratories. Owing to a result in quantum information theory known as the gentle measurement lemma, these observables describe conditional localization probabilities. In principle, these observables satisfy commutativity when associated with causally separated laboratories. As a consequence, they could be represented by local observables in the sense of Araki-Haag-Kastler. Explicit examples of such local observables satisfying the commutativity requirement are presented in a second paper within the framework of local quantum field theory.
\end{abstract}
\tableofcontents

\section{Introduction} 
 \subsection{Localization, Causality, Locality, Commutativity}
{\bf Newton-Wigner  localization and first issues with causality}. Spatial localization in relativistic quantum theory has long been a delicate issue, lying at the intersection of quantum measurement theory, spacetime causality, and the very notion of particle.

In non-relativistic quantum mechanics, localization is naturally described by projection-valued measures and, correspondingly, by selfadjoint position operators. 
This is a case of {\em sharp localization}, since the projectors $P(\Delta)$ of the spectral measure are labelled by arbitrarily small regions $\Delta$ of the rest space $\Sigma$ of an observer, and we can prepare states $\rho$ for which the probability of detecting a particle in such a bounded region is $tr(\rho P(\Delta))= 1$.

In the relativistic setting, however, this scheme encounters severe difficulties. The best-known example is the Newton-Wigner construction \cite{NW}. This notion of localization is given in terms of a triple of selfadjoint position observables associated with every rest space of a given inertial observer and with a given elementary particle in Wigner's sense (an irreducible positive-energy strongly continuous representation of the orthochronous inhomogeneous Poincar\'e group). Later, Wightman \cite{Wig} proved a uniqueness theorem based on Mackey's imprimitivity theory for the Euclidean group and some (actually awkward) regularity conditions. 

Newton-Wigner localization was proposed as a possible tool for reinterpreting localization in (free) quantum field theory so as to avoid some of the puzzling consequences related to the {\em Reeh-Schlieder property} \cite{Fleming}. Nevertheless, it was later shown that the counterintuitive features implied by the Reeh-Schlieder theorem persist in this reformulation \cite{H1}.

From the one-particle viewpoint, a further serious {\em causality} issue is raised by Hegerfeldt's celebrated analysis \cite{Hegerfeldt,Hegerfeldt2}. If energy, i.e., the selfadjoint generator of time translations, is positive (or bounded from below), sharp localization described by projection-valued measures gives rise to specific causality problems. In modern terms, and especially in the formulation later clarified by Castrigiano, these problems concern the incompatibility of sharp localization with {\em causal propagation} of localization probabilities: in the cases considered, localization probabilities are found to evolve superluminally (and could be used to propagate macroscopic information superluminally \cite{M23}).\\

\noindent{\bf Unsharp localization}.
 This strongly suggests that, in relativistic quantum theory, sharp notions of localization should be abandoned in favour of unsharp ones described by POVMs. In this case, {\em effects}, that is, positive operators bounded by the identity, $0\leq A(\Delta)\leq I$, rather than orthogonal projectors, are still associated with arbitrarily small regions $\Delta$ in the rest space $\Sigma$ of an observer, but the probability of finding the particle in a given bounded region is strictly $tr(\rho A(\Delta))<1$ for every state $\rho$, even though sharp states can be approximated arbitrarily well.

A sequence of works, following a road map initiated by Castrigiano \cite{Castrigiano2,M23,C23,DM24,C24,CDRM,DRCM26}, has shown that this strategy is in fact viable. A broad family of relativistic spatial localization observables has been constructed and analyzed for various elementary systems, with full compatibility with Poincar\'e covariance and with natural causal requirements, especially in the unsharp setting under the positive-energy requirement.
These developments indicate that Hegerfeldt's issue does not rule out unsharp localization under the requirement of energy positivity, but rather rules out an excessively rigid, projection-based notion of localization.\\

%%%%%%%%. METTERE COMMENTO SULLA LOCALIZAZIONE MODULARE 

\noindent {\bf Modular localization}. Newton--Wigner localization may nowadays be considered untenable, in view of its causality problems, when compared with the unsharp notion considered in this paper. However, another interesting approach to relativistic localization is provided by {\em modular localization}, which partially 
reintroduces some features of the Newton--Wigner approach. Early developments of this framework can be found in the work of Schroer \cite{SchroerML}, while its systematic formulation at the single-particle level, in terms of Wigner representations, was achieved in \cite{BGL}. In this framework, spacetime regions are associated with closed {\em real} subspaces of the complex one-particle Hilbert space. It is therefore still based on projections (rather than effects), but on {\em real}-linear orthogonal projections onto {\em real} subspaces, rather than on the usual complex-linear projections entering PVMs and the Born rule. A recent analysis by Lechner and de Oliveira \cite{LdO} develops this perspective into a causal and Poincar\'e-covariant localization scheme with values in the lattice of real projections.
Evaluating the resulting localization structure in a state gives rise to a {\em fuzzy probability measure}, which is generally nonadditive and therefore is not a genuine probability measure. Nevertheless, for separation scales large compared with the Compton wavelength, the corresponding modular localization scheme becomes essentially additive and, remarkably, approximates Newton--Wigner localization.
The relation between this alternative point of view and the POVM-based localization studied here deserves further investigation and will be pursued elsewhere.\\

\noindent {\bf Localization and relativistic locality in terms of commutativity}. Returning to localization schemes described by families of POVMs and to their causality problems,  
 %%%%%%%%%%%%%%%%%%%
 there is, however, a {\em second} and subtler causality problem, which concerns not the propagation of localization probabilities in time, but rather the representation of {\em relativistic locality}  itself in measurement theory. A central result in this direction is the no-go theorem by Halvorson and Clifton \cite{HC}, which may be regarded as the culmination of a line of investigation initiated by Schlieder \cite{S} and Jancewicz \cite{J}, sharpened by Malament \cite{Malament} for projector-valued localization, and later extended by Busch \cite{Buschloc}. Roughly speaking, Halvorson and Clifton proved that, under a family of assumptions that appear very general and natural for a relativistic localization scheme, the localization effects (elements of an alleged localization POVM) must be trivial. In this sense, the theorem seems to suggest a deep tension between relativistic localization and relativistic locality.
We stress that the theorem does not assume any particular definition of the quantum system under consideration, which, in principle, could be formalized as a Wigner elementary particle, a composite system, or even a quantum field. The assumptions concern additivity, translation covariance, positivity of energy, and a {\em microcausality} requirement expressed by commutativity of localization effects associated with suitably ``separated'' regions $\Delta,\Delta'$:
$$[A(\Delta),A(\Delta')] =0\:.$$
Actually, it is not even assumed that a full unitary representation of the Poincar\'e group exists. The relativistic perspective comes into play only to justify physically the commutativity of localization effects associated with ``separated'' regions, when these are interpreted as {\em causally} separated regions in the relativistic sense.

The crucial point is that the commutativity requirement entering the Halvorson-Clifton theorem is usually regarded as self-evident from the perspective of local quantum physics, especially within the Araki-Haag-Kastler framework \cite{Haag,Araki}, where locality is encoded precisely by commutativity of observables associated with causally separated spacetime regions. If the spatial region $\Delta$ of a rest space $\Sigma$ where a particle is found at a certain time $t$ belongs to the spacetime region ${\cal O}$ where the localization observable $A(\Delta)$ is localized (in the Araki-Haag-Kastler sense),
then the commutativity assumption would indeed appear unavoidable. From this standpoint, Halvorson-Clifton's result would represent an insurmountable obstruction to any notion of position of a quantum system at a given time in a reference frame. However, this line of reasoning already presupposes a specific mathematical representation of relativistic locality. It is therefore natural to investigate the interplay of commutativity and causal separation of the localization region, independently of the assumption that localization observables $A(\Delta)$ belong, in the strict sense, to Araki-Haag-Kastler local algebras $\gA({\cal O})$ attached to arbitrarily small spacetime neighborhoods ${\cal O}$ of the corresponding spatial regions: ${\cal O}\supset \Delta$. A first step in this direction was taken by Busch \cite{Buschloc}, building on a technical result with Singh \cite{BS}, by analyzing the relation between commutativity and relativistic locality directly within the general theory of quantum measurement.
In Busch's approach, an observable is localized in a spacetime region ${\cal O}$ if it can be measured with instruments operating there. This is also in the spirit of the Araki-Haag-Kastler framework; in contrast to that framework, however, no further requirement concerning the structure ($*$, $C^*$, or von Neumann algebra) of the family of observables localized in a spacetime region is imposed. Similarly, no {\em a priori} requirements are imposed on families of observables and operations associated with distinct spacetime regions.
In that setting, relativistic locality is formulated through operational principles such as the {\em no-signaling} condition and, in a stronger form, what is often called {\em relativistic consistency} \cite{tesi}, namely invariance of the statistics of suitable consecutive measurements under interchange of their temporal order when the measurements take place in causally separated regions. For broad classes of measurement procedures -- in particular for L\"uders measurements and related efficient schemes -- these principles imply commutativity of the relevant effects. This analysis is mathematically precise and conceptually important because it does not {\em assume} relativistic locality in the algebraic form of microcausality (i.e. commutativity), but rather attempts to {\em derive} commutativity from independent operational requirements.\\
 
\noindent {\bf Purposes of this work}. The main purpose of this paper is to investigate whether this approach, initiated by Busch, can actually justify the commutativity requirement for relativistic spatial localization observables. The answer we shall propose is, in general, negative. To this end, following Busch's approach, we adopt in particular a natural {\em local detectability} principle, which asserts that if instruments operate in an open spacetime region ${\cal O}$ and are used to detect a particle in a spatial region $\Delta$, then $\Delta \subset {\cal O}$.

The above negative answer is not due to a failure of no-signaling or relativistic consistency, but rather to the physical meaning of localization for a quantum system which we define to be a {\em particle}. Let us assume that a localization observable is meant to describe a system with the propensity to be found in a unique place when an ideal position measurement is performed on a whole rest space.
Then the elementary observable associated with detection in a region $\Delta$ is not naturally localized in an arbitrarily small spacetime neighborhood of $\Delta$.
On the contrary, once one adopts the above-stated local detectability principle, as we shall prove in Section \ref{NCn}, one is led to the conclusion that the elementary observable $\{A(\Delta),I-A(\Delta)\}$ is localized, in Busch's sense, in a region that must contain the whole rest space where the global localization experiment is performed. As a consequence, two elementary localization observables associated with spatially disjoint regions are not localized in causally separated spacetime regions, and therefore the no-signaling and relativistic-consistency arguments leading to commutativity cannot even be triggered. In this sense, the commutativity requirement for localization effects does {\em not} follow from those more elementary locality principles, since the very hypotheses for employing these principles are not satisfied.
This conclusion is physically tied to the {\em particle character} of the systems under consideration. A particle is assumed not to localize simultaneously in several (causally) disjoint regions in an ideal exhaustive detection experiment, but rather in a unique place. If one measures whether the particle is found in a region $\Delta$ within a complete position measurement on a rest space $\Sigma$, then one automatically obtains information about the complementary region $\Sigma\setminus \Delta$. From this viewpoint, the associated elementary observable is not confined to $\Delta$ alone: it carries information about the global detection arrangement on $\Sigma$. The failure of the commutativity argument is therefore not an exotic mathematical pathology, but rather a direct consequence of the physical meaning attributed to localization of a particle. If we do not accept this, we cannot accept the very notion of a particle.

These arguments support the claim that localization effects for particles do not belong to local algebras in the sense of the Araki-Haag-Kastler formulation when they are effects of a POVM defined on the whole rest space of a reference frame (a Cauchy surface).

The second aim of the paper is to indicate a possible way to make spatial localization and commutativity compatible in more realistic localization experiments.

We introduce a notion of {\em conditional localization} associated with a bounded spatial region $\Delta_0$, interpreted as the spatial part of a laboratory. The idea is to construct, starting from a general localization observable defined on a full rest space, a new POVM normalized on $\Delta_0$ which describes the  probability of finding the system in a subregion $\Delta\subset \Delta_0$, given that the system has been detected in $\Delta_0$.
The mathematical possibility of such a construction relies on a known result in quantum information theory called the {\em gentle measurement lemma}. It permits one to prove that the specific class of POVMs introduced here, normalized on a bounded spatial region $\Delta_0$ representing a laboratory, measures the conditional probability of finding a particle in $\Delta\subset \Delta_0$ if we know that the system was detected in $\Delta_0$.
This is all the more true the closer the probability of finding the particle in $\Delta_0$ is to $1$. These conditional POVMs are intrinsically local to the laboratory and, unlike the original localization observables on the whole rest space, are not affected by the same problem that prevents the use of no-signaling or relativistic-consistency arguments. We shall argue that pairs of such conditional POVMs associated with causally separated laboratories may, in principle, satisfy commutativity and belong to local algebras concentrated around the laboratory region in the Araki-Haag-Kastler sense. This suggests that conditional localization may provide a mathematically and physically appropriate notion of local position measurement in relativistic quantum theory.

Explicit examples of these local observables, satisfying a commutativity requirement and constructed in terms of local energies, are presented in a second paper \cite{M26} within the framework of local QFT.

\subsection{Structure of the work}
The paper is organized as follows. In Section 2 we review the notion of relativistic spatial localization observable and the relevant causal conditions, and we recall the Halvorson-Clifton no-go theorem for unsharp localization systems. We also discuss the relation of the problem with the Araki-Haag-Kastler framework. In Section 3 we summarize Busch's operational analysis of locality in terms of no-signaling and relativistic consistency, including the relevant results connecting these principles with commutativity of effects. In Section 4 we apply that analysis to relativistic localization observables and prove that, under a natural local detectability principle and for systems with particle-like localization behaviour, commutativity of causally separated localization effects is not implied by those operational locality principles. In Section 5 we introduce conditional localization POVMs in bounded laboratories, study their basic mathematical properties, and discuss why they may evade the previous problem. The paper closes with conclusions and an outlook toward a subsequent work, where explicit examples of localization observables and conditional local position observables will be constructed within a local quantum-field-theoretic framework.

\subsection{Notions and notation}   
We assume $c=1$ and $\hbar=1$ throughout the rest of this paper; the notation $A\subset B$ includes the case $A=B$. All Hilbert spaces used in this work are complex and {\em separable}.\\  
  
\noindent {\bf Minkowski spacetime}.  
A four-dimensional real affine space whose space of translations $\sV$ is equipped with a non-degenerate symmetric bilinear form $g$ with signature $(-,+,+,+)$ is called the {\bf Minkowski spacetime} $\bM$. The points of $\bM$ are called {\bf events}, and $g$ is called the {\bf Minkowski metric}.    
%We shall make use of the notation $u\cdot v :=g(u,v)$ if $u,v \in \sV$.  We also use the dot to indicate the standard (positive) scalar product of $3$-vectors $\vec{u}\cdot \vec{v}$ viewing them as spacelike vectors (see below).  
%Upon the choice of an origin $o\in \bM$, the points $p\in \bM$ are one-to-one with vectors of $\sV$ through the map $\bM \ni p \mapsto p-o \in \sV$. We shall take advantage of this identification several times in the rest of this paper.  If $p\in \bM$ and $v\in \sV$, $q=p+v$ means $q-p=v$.   
A vector $v\in \sV$  
is {\bf spacelike} if $g(v,v)>0$ or $v=0$.  
It is {\bf causal} if $g(v,v)\leq 0$ and $v\neq 0$. A causal vector $v$ is {\bf timelike} if $g(v,v) <0$,  
or {\bf lightlike} if $g(v,v) =0$.  
Smooth curves are classified analogously according to their tangent vectors, provided their causal type is constant along the curve.
   
A set $\Lambda\subset \bM$ is {\bf achronal} if $p-q$ cannot be timelike for $p,q\in \Lambda$, and {\bf spacelike} if $p-q$ is spacelike for $p\neq q$.   
%A {\bf maximal achronal set} is an achronal set which is not a proper subset of a similar set.  $\Lambda \subset \bM$ is {\bf spacelike} if  $p-q$ spacelike  
% for $p,q\in \Sigma$.  
  
The set of timelike vectors is an open cone consisting of two disjoint open connected components. A choice of one of them $\sV_+$ defines a {\bf time orientation} of  
$\bM$. The latter is henceforth assumed to be {\bf time oriented}: $\sV_+\subset \sV$ is the open cone of {\bf future-directed timelike vectors}. $\overline{\sV_+} \setminus\{0\}$ is the cone of {\bf future-directed causal vectors}. Notice that if $v \in \overline{\sV_+}$ and $u$ is causal, then $u \in \overline{\sV_+}$ if and only if $g(u,v) \leq 0$.  
We finally define $\sT_+ := \{u\in \sV_+ \:|\: g(u,u)=-1\}$.

If $A\subset \bM$, the {\bf causal future} $J^+(A):= \{p\in \bM \:|\: p-q \in \overline{\sV_+} \quad \mbox{for some $q\in A$}\}$ represents the events of $\bM$ in the future of $A$ that can be physically influenced by $A$.  
The {\bf causal past} $J^-(A):= \{p\in \bM \:|\: q-p \in \overline{\sV_+} \quad \mbox{for some $q\in A$}\}$ is defined symmetrically.  
  
$A,B \subset \bM$ are said to be {\bf causally separated} if $(J^+(A)\cup J^-(A)) \cap B= \emptyset$ (which is equivalent to $(J^+(B)\cup J^-(B)) \cap A= \emptyset$).  
 
 If $R\subset \bM$, its {\bf causal complement} and its {\bf causal completion}, respectively, are \beq R^\pperp := \bM \setminus((J^+(R)\cup J^-(R))) \quad \mbox{and}\quad(R^{\pperp})^{\pperp} \:.\label{CComp}\eeq   It is easy to prove that $R\subset (R^{\pperp})^{\pperp}$. If $R\subset \Sigma$, the latter being a flat spacelike 3-plane, then $(R^{\pperp})^{\pperp}$ turns out to consist of the points $p\in \bM$ such that every causal {\em straight line} passing through $p$ meets $R$ somewhere\footnote{  
This is equivalent, in $\bM$, to the set of points $p$ such that every inextensible causal curve passing through $p$ also meets $R$ somewhere. This latter set, in a generic spacetime $M$, is also called the {\em domain of dependence} of $R$ (usually when $R$ is also assumed to be achronal).}.

A {\bf Minkowskian reference frame} -- physically representing an {\bf inertial reference frame} or an {\bf observer} -- is a unit timelike vector $n\in \sT_+$.

A {\bf rest space} of a Minkowski reference frame $u\in \sT_+$ is an affine $3$-plane $\Sigma\subset \bM$ $g$-normal to $u$, written $u\perp \Sigma$.  Notice that $\Sigma$ is therefore necessarily spacelike.
If $o\in \bM$ is a given origin, the family of rest spaces of $u\in \sT_+$ is labeled by the time at which they occur in the said reference frame. 
Every spacelike affine $3$-plane $\Sigma$ is, at some time, the rest space of a Minkowskian reference frame, which is then said to be {\bf adapted} to $\Sigma$. 
   
$d\Sigma$ denotes the natural translationally invariant Borel measure on the rest space $\Sigma$ of $u$: the Lebesgue measure on $\Sigma$ induced by the metric. 
$\cB(\Sigma)$ denotes the family of Borel subsets of $\Sigma$.  
We use the notation $|\Delta| = \int_\Sigma \chi_\Delta(x) d\Sigma(x)$  
%= \int_{\bR^3} \chi_\Delta(x) d^3x$   
if $\Delta \in \cB(\Sigma)$.

The Lie group of metric-preserving affine maps $h: \bM \to \bM$ is known as the {\bf Poincar\'e group} $IO(1,3)$. The Lie subgroup of affine maps that also preserve the time orientation is called the {\bf orthochronous Poincar\'e group} $IO(1,3)_+$.   
The subgroup of $IO(1,3)_+$ that leaves fixed an arbitrarily chosen origin\footnote{Different choices of $o$ give rise to isomorphic definitions and the component $\Lambda$ of an element of $IO(1,3)_+$ does not depend on the choice of $o$.} $o\in \bM$ is the {\bf orthochronous Lorentz group} $O(1,3)_+$.  
Elements $h\in IO(1,3)_+$ are in one-to-one correspondence with pairs $(\Lambda, v)$ where $v\in \sV$ and $\Lambda \in O(1,3)_+$, and the action of $h$ on a point $q\in \bM$ is  
\beq (\Lambda, v) q = o + v+ \Lambda (q-o)\:. \eeq

The subgroup of $IO(1,3)_+$ known as the {\bf proper orthochronous Poincar\'e group} $ISO(1,3)_+$ is obtained by replacing $O(1,3)_+$ with the {\bf proper orthochronous Lorentz group} $SO(1,3)_+$.  
The latter, representing $O(1,3)_+$ in a Minkowskian coordinate system as above, is obtained by restricting to Lorentz matrices with $\det \Lambda >0$.   
  
For some physical systems, such as fermions, Poincar\'e invariance is implemented in terms of {\em projective unitary representations}   
$V: ISO(1,3)_+ \to \gB(\cH)$ that become genuinely unitary when viewed as   
representations of the universal covering of the {\bf proper orthochronous Poincar\'e group} $\widetilde{ISO(1,3)_+}$. In other words (see, e.g., \cite{Moretti2}), the representation  
$U:= V \circ \pi : \widetilde{ISO(1,3)_+} \to \gB(\cH)$  
is unitary if $\pi: \widetilde{ISO(1,3)_+} \to ISO(1,3)_+$ is the canonical covering homomorphism.  
This result is obviously valid even if the initial representation $V$ is genuinely unitary. As a consequence, we sometimes replace $ISO(1,3)_+$ with $\widetilde{ISO(1,3)_+}$ in some definitions in order to cover a larger class   
of physical situations.\\  
  
\noindent {\bf Operators and observables}.  
An operator $A$ on a Hilbert space $\cH$ has its own domain $D(A)\subset \cH$ and is written $A:D(A)\to\cH$. Combinations and compositions of operators $A+B$, $AB$, and $aA$ ($a\in \bC$) are defined on their {\bf standard domains}  
$D(A+B):= D(A)\cap D(B)$, $D(AB):= \{x\in D(B) \:|\: Bx \in D(A)\}$, $D(aA):=D(A)$ unless $a=0$, for which $D(0A):=D(0) := \cH$.  
(The standard domain $D(AB)$ is defined analogously when the domains and ranges of $A$ and $B$ are subspaces of different Hilbert spaces, provided the composition $AB$ makes sense.) The adjoint operator is denoted by $A^\dagger: D(A^\dagger) \to \cH$.

$\gB(\cH)$ denotes the $C^*$-algebra of bounded operators $A:\cH\to \cH$ for a Hilbert space $\cH$. If $A,B \in \gB(\cH)$, $A\geq B$ (equivalently $B\leq A$) means that $A-B\geq 0$; in turn, 
$A\geq 0$ means that $\langle \psi|A\psi \rangle \geq 0$ if $\psi\in \cH$. In that case we say that the operator $A$ is {\bf positive}.

The (generalized) notion of observable that we shall use throughout is that of a {\bf Positive Operator-Valued Measure} (POVM) on a Hilbert space $\cH$. It is a map   
$\cF(X) \ni B \mapsto E(B)$, where $\cF(X)$ is a $\sigma$-algebra on the set $X$, each $E(B)\in \gB(\cH)$ is an {\bf effect}, i.e., it satisfies $0\leq E(B)\leq I$, together with the {\bf normalization condition} $E(X)=I$, and, finally, the requirement that  
for every $\psi \in \cH$, the associated map $\cF(X) \ni B \mapsto \langle \psi|E(B)\psi\rangle$ is $\sigma$-additive and therefore defines a positive measure on $X$, which is a probability measure if $||\psi||=1$. Due to the positivity of the operators involved, this condition is equivalent to the strong $\sigma$-additivity of the map $\cF(X) \ni B \mapsto E(B)$, which, obviously, is also additive.  
$X$ is interpreted as the set of outcomes of the observable described by the POVM $\cF(X) \ni B \mapsto E(B)$.

Generally, mixed states $\rho$ are trace-class operators $\rho \in \gB_1(\cH)$, positive ($\rho \ge 0$), and normalized ($\mathrm{tr} \rho = 1$). The convex body of states will be denoted by $\sS(\cH)$.  
A special case of states is given by one-dimensional projectors $\rho = |\psi\rangle\langle \psi|$ for unit vectors\footnote{These are pure states, i.e., extremal elements in the space of states if the von Neumann algebra of observables is the whole $\gB(\cH)$.}, $\psi \in \cH$.   
  
For a state $\rho\in \sS(\cH)$, $tr(E(B)\rho)= tr(\rho E(B))$ is interpreted as the probability of obtaining an outcome in $B$ when the system is in the state $\rho$.  
  
If $\cF(\bR)=\cB(\bR)$ and all the effects $E(B)$ are orthogonal projectors, we have a standard {\bf Projector-Valued Measure} (PVM). As is well known, every PVM is in one-to-one correspondence with a selfadjoint operator $\hat{E}= \int_{\bR} xE(dx)$ through the spectral theorem. In this sense, a POVM is a generalized observable.  
  
In the special case where $\cF(X)$ is the power set of $X:= \{1,\ldots, N\}$, a POVM on $X$ is completely determined by the effects $E_j := E(\{j\})$ with $j\in \{1,\ldots, N\}$.  
As a general textbook reference for this mathematical formalism and its applications to physics, we recommend \cite{Buschbook}. The references on general spectral theory applied to physics that we shall use are \cite{Moretti1,Moretti2}.  
 
\section{Halvorson--Clifton's no-go result for unsharp localizations}  
In what follows, we first introduce some types of spatial localization, based on the POVM formalism, which are compatible with Hegerfeldt's locality constraint. We shall simultaneously review Hegerfeldt's causality conditions in the modern presentation due to Castrigiano.   
Afterwards, we shall present the Halvorson--Clifton no-go result \cite{HC}. Finally, we shall focus on the crucial {\em commutativity} hypothesis assumed by Halvorson and Clifton from the standpoint of Busch's analysis of locality \cite{Buschloc} and of more recent developments on the subject, especially those collected in Beck's thesis \cite{tesi}.

\subsection{Types of unsharp spatial localization} \label{SECUNS}

We now consider a relativistic quantum system in Minkowski spacetime $\bM$ described in a complex Hilbert space $\cH$.  
An unsharp notion of localization in the region $\Delta$ of the rest space $\Sigma$ of a Minkowski reference frame $u\in \sT_+$ at a given time in $\bM$ is represented by an {\em effect} $A_0(\Delta) \in \gB(\cH)$. Therefore  $\mathrm{tr}(A_0(\Delta)\rho)$ (or $\langle \psi|A_0(\Delta)\psi\rangle$ for vector states) is the probability of finding the system in $\Delta$ -- at the instant of time defining $\Sigma$ -- when the pre-measurement state is $\rho\in \sS(\cH)$.  
We also expect that spacetime translations of the employed instruments are represented by some unitary representation of the translation group $\sV$ of $\bM$. 
Based on these minimal principles, we can formally define, according to \cite{Buschloc,HC}:\\ 
 
\begin{definition}[Unsharp Localization System]\label{UNSH} 
An {\bf unsharp localization system} in Minkowski spacetime $\bM$ is a quadruple $(\cH, {\cal R}_0, A_0, U_0)$ where: 
\begin{itemize} 
\item[(a)] $\cH$ is a complex Hilbert space; 
\item[(b)] $A_0: \mathcal{R}_0 \ni \Delta \mapsto A_0(\Delta) \in \gB(\cH)$ is an effect-valued map ($0 \le A_0(\Delta) \le I$), with $\mathcal{R}_0$ a family of bounded sets of a foliation of parallel spacelike  
3-planes in $\bM$; 
\item[(c)] $U_0: \sV \ni a \mapsto U_{0a}\in \gB(\cH)$ is a strongly continuous unitary representation of the translation group $\sV$ of the affine space $\bM$. \hfill $\blacksquare$ 
\end{itemize}  
\end{definition} 
 
The above structure is sufficient for stating the most basic results. However, for the development of our discussion we shall need a more elaborate construction, which specializes Definition \ref{UNSH} and has also been used in recent literature.  
From a physical perspective, the mathematical machinery introduced in Definition \ref{UNSH} does not seem sufficient to capture the physical requirements embodied in an (idealized) exhaustive detection experiment. If we want to interrogate a quantum system about its position in space, we should force it to take a position in space. Assuming that a detection is an interaction with an instrument, this means that we have to fill the {\em entire} space with detectors. Within this setting, the system localizes somewhere at the given time. 
If detectors are mathematically represented by effects $A_0(\Delta)$, we are forced to assume that either  
infinitely extended regions $\Delta$ are permitted, or the map $\Delta \mapsto A_0(\Delta)$ is $\sigma$-additive (and not simply additive), or both, in order to make contact with the standard formulation of probability theory. Another observation is that the relativity principle suggests covariance with respect to the whole orthochronous Poincar\'e group. 
According to these ideas, we can state the following general definition of a {\em relativistic spatial localization observable}, adapted\footnote{In those references, the POVMs on the various $\Sigma$ were defined independently with a further coherence condition, and an absolute continuity condition was added in \cite{DM24}.} from \cite{M23,DM24}. This notion, in a slightly simplified version, was first introduced and thoroughly analyzed by Castrigiano (see \cite{Castrigiano2} and references therein) under the name of {\em Poincar\'e-covariant POL}.\\ 
 
\begin{definition}[Relativistic Spatial Localization Observable]\label{REMM0} 
A {\bf relativistic spatial localization observable} is a quadruple $(\cH, {\cal R}, A, U)$ where  
\begin{itemize} 
 \item[(a)] ${\cal R}= \cup \{\cB(\Sigma)\:|\: \Sigma \subset \bM \: \mbox{spacelike  $3$-plane}\}$ where $\cB(\Sigma)$ 
is the Borel $\sigma$-algebra on $\Sigma$; 
    \item[(b)]  $A= \{A^v\}_{v\in \sS}$ is a family of  maps $A^v: {\cal R} \to \gB(\cH)$ -- where $\sS$ is a set of tensors of definite order which is invariant under $O(1,3)_+$ -- 
 such that every restriction 
    $A^v\spa \rest_{\cB(\Sigma)} : \cB(\Sigma) \to \gB(\cH)$ is a (normalized) POVM; 
\item[(c)]  $U: IO(1,3)_+ \to \gB(\cH)$  is a  strongly-continuous unitary representation of the orthochronous Poincar\'e group; 
\item[(d)] $A$ is $U$-covariant, i.e.,  $U_g A^v(\Delta) U_g^{-1} = A^{gv}(g\Delta)$ for every $g\in {IO(1,3)_+}$ and $\Delta \in {\cal R}$, 
 where  
$\sS \ni v\mapsto gv\in \sS$ denotes the action of  ${IO(1,3)_+}$ on the tensors in $\sS$. \end{itemize} 
 
In (a), possibly $\sS=\{1\}$, so that $A$ includes a unique element denoted by the same symbol $A$. In this case the relativistic spatial localization observable is said to be {\bf scalar}. 
 
In (c), $U$ may be generalized to a strongly-continuous unitary representation of the universal covering of the proper orthochronous Poincar\'e group $\widetilde{ISO(1,3)_+}$ provided $g\Delta$ and $gv$ in (d) are respectively replaced  
by $\pi(g)\Delta$ and $\pi(g) v$ -- where $\pi: \widetilde{ISO(1,3)_+}  \to  {ISO(1,3)_+} $ is the canonical covering homomorphism.\\ 
\end{definition}

Consider a relativistic spatial localization observable $(\cH, {\cal R}, A, U)$.  
Define the subfamily ${\cal R}_0\subset {\cal R}$ containing the bounded elements of ${\cal R}$ that lie in a chosen foliation of parallel spacelike $3$-planes, and let $U_0$ be the restriction of $U$ to the subgroup of spacetime translations. In this case 
$(\cH, {\cal R}_0, A_0:=A^v\spa\rest_{{\cal R}_0}, U_0)$ 
is an unsharp localization system for every given $v\in \sS$. In this sense, the notion of relativistic spatial localization observable is an {\em extension} of the notion of unsharp localization system. 
In both definitions, the effects $A^{v}(\Delta)$ are allowed to be orthogonal projectors, and some (or all) of the considered POVMs may be PVMs.

At this stage, we can state, in modern language, the causality conditions implicit in Hegerfeldt's original analysis \cite{Hegerfeldt, Hegerfeldt2} concerning the impossibility of superluminal propagation of detection probability.\\

\begin{definition}[Castrigiano's causality conditions]\label{REMM0444} 
{\em A relativistic spatial localization observable $(\cH, {\cal R}, A, U)$ is {\bf causal} if it satisfies the {\bf  causal condition}: for every $\Delta \in {\cal R}$, every spacelike $3$-plane $\Sigma$ and every $v\in \sS$ 
\begin{itemize} 
    \item[{\bf (CC)}] $\quad$ $A^v(\Delta) \leq A^v(\Delta_\Sigma) \quad \mbox{where $\Delta_\Sigma:= \Sigma \cap (J^+(\Delta)\cup J^{-}(\Delta))$ provided $\Delta_\Sigma \in {\cal R}$.}$ 
  \end{itemize} 
A relativistic spatial localization observable $(\cH, {\cal R}, A, U)$ satisfies the {\bf causal time condition} with respect to a Minkowskian reference frame $u\in \sT_+$\footnote{The original notion of CT \cite{Castrigiano2} does not assume Poincar\'e covariance and refers to a preferred notion of time evolution. Assuming covariance, our definition immediately arises from the original one.} if, for every pair of spacelike $3$-planes $\Sigma_0, \Sigma$ normal to $u$ and 
$\Delta \in {\cal R}$ with  
$\Delta \subset \Sigma_0$, and every $v\in \sS$, we have  
\begin{itemize} 
    \item[{\bf (CT)}] $\quad$ $A^v(\Delta) \leq A^v(\Delta_\Sigma) \quad \mbox{where $\Delta_\Sigma:= \Sigma \cap (J^+(\Delta)\cup J^{-}(\Delta))$ provided $\Delta_\Sigma \in {\cal R}$.}$ \hfill $\blacksquare$ 
    \end{itemize}}  
\end{definition}

\noindent Evidently, CC $\Rightarrow$ CT, but the converse implication is generally false (see \cite{M23}).\\ 
 
\begin{remark}  
{\em\begin{itemize} 
\item[(1)] It follows from Theorem (9) in \cite{Castrigiano2} that every relativistic spatial localization observable defined on a {\em separable} Hilbert space\footnote{As is the case for the Hilbert spaces encountered in this work.}, solely by virtue of covariance with respect to the translation group of a spatial $3$-plane $\Sigma$, satisfies $A^v(\Delta)=0$, for $\Delta \in \cB(\Sigma)$, if and only if $|\Delta|=0$. In particular, $\cB(\Sigma)\ni \Delta \mapsto \langle\psi| A^v(\Delta)\psi \rangle$ is absolutely continuous with respect to the Lebesgue measure $d\Sigma$ on the $3$-plane $\Sigma$ for every given $\psi\in \cH$. 
 
    \item[(2)] As a consequence of (1), every $A^v|_{\cB(\Sigma)}$ can be uniquely completed to a POVM defined on the Lebesgue $\sigma$-algebra $\cL(\Sigma)$. 
 
\item[(3)] If one completes ${\cal R}$ by including the Lebesgue sets, we have  
$\Delta_\Sigma \in {\cal R}$ if $\Delta \in {\cal R}$, by Lemma (16) in \cite{Castrigiano2}. With this extension, the requirement ``provided $\Delta_\Sigma \in {\cal R}$'' in (CC) and (CT) can be omitted. \hfill $\blacksquare$\\ 
\end{itemize}} 
\end{remark} 
 
\noindent We stress that  
CC includes the original causality condition discussed by Hegerfeldt in his celebrated works \cite{Hegerfeldt, Hegerfeldt2}, which, in Castrigiano's modern setting, corresponds to CT. Hegerfeldt proved that CT is violated by relativistic spatial localization observables described in terms of PVMs, in particular Newton--Wigner localization. This puzzling result was established under the hypothesis that the generators of timelike translations have spectra bounded from below \cite{Castrigiano2, M23}. 
Violation of CT implies violation of CC. In summary, if we assume CC, we are forced to {\em abandon every sharp notion of localization} described by PVMs in order to comply with causality.
The explicit examples of relativistic spatial localization observables presented in 
\cite{Castrigiano2, M23, DM24, C23, C24, CDRM, DRCM26} for various types of particles\footnote{These relativistic spatial localization observables, when associated with {\em Wigner elementary particles}, have positive selfadjoint generators of time translations.} show that Hegerfeldt's causality issue can in fact be considered harmless. This is because CC is satisfied provided one uses suitable unsharp localizations in terms of POVMs as in Definition \ref{REMM0}. 
  
The notion of unsharp localization in Definition \ref{REMM0} has been extended further by allowing rest spaces  
$S$ that are less rigid than the planes $\Sigma$. They can be chosen as {\em smooth Cauchy surfaces} (especially spacelike ones) \cite{DM24} or even (Lipschitz) {\em maximal achronal sets} 
of $\bM$ \cite{C24}, not necessarily Cauchy surfaces. Adopting these more general settings, the (generalized) causal condition CC turns out to be {\em equivalent} \cite{DM24,C24} to the normalization of the POVMs defined on all (generalized) rest spaces. 
As an important byproduct, it was recently proved in \cite{CDRM,DRCM26} that extending the formalism to achronal sets leads to a representation of the so-called {\em causal logic of spacetime} in terms of effects for both massive bosons and massive fermions.

\subsection{The no-go result by Halvorson and Clifton} 
Although Hegerfeldt's causality issue seems to be solved by using a suitable unsharp notion of localization, a second type of causality issue remains. It concerns the representation of locality in terms of commuting operators associated with observables localized in causally separated regions. The most modern version of this issue was formulated in the famous   
Halvorson-Clifton paper \cite{HC}, extending Schlieder's \cite{S}, Jancewicz's \cite{J},
and 
Malament's \cite{Malament} arguments (valid for sharp localization) to the case of unsharp localization.\\ 
 
\begin{theorem}[Halvorson-Clifton] \label{HC} 
Let $(\cH, {\cal R}_0, A_0, U_0)$ be an unsharp localization system in Minkowski spacetime satisfying the following properties: 
\begin{itemize} 
\item[1.] {\bf Additivity}: ${\cal R}_0$ is closed with respect to finite unions of regions in the same 3-plane, and $A_0$ is additive thereon: $$A_0(\Delta) + A_0(\Delta') = A_0(\Delta \cup \Delta') \quad \mbox{for $\Delta, \Delta' \subset \Sigma$ with $\Delta\cap\Delta'=\emptyset$ and $\Delta,\Delta'\in {\cal R}_0$.}$$ 
\item[2.] {\bf Translation covariance}: ${\cal R}_0$ is closed under translations of $\sV$, namely $\Delta_{v} := \Delta + {v} \in {\cal R}_0$ if ${v} \in \sV$ and $\Delta \in {\cal R}_0$, and 
$$U_{0v} A_0(\Delta) U^{-1}_{0v} = A_0(\Delta + {v}).$$ 
\item[3.] {\bf Energy bounded below}: If ${u} \in \sV_+$, the self-adjoint generator $H_{0 u}$ of the strongly continuous unitary one-parameter group 
$$\bR \ni t \mapsto U_{0t{u}} = e^{-i t H_{0u}}$$ 
has spectrum bounded from below. 
\item[4.] {\bf Microcausality}: Let $\Delta, \Delta' \in {\cal R}_0$ lie in a common spacelike 3-plane $\Sigma \subset \bM$, and suppose they are disjoint with strictly positive distance. If ${u} \in \sV_+$, there exists $\delta > 0$ such that  
$$[A_0(\Delta), A_0(\Delta' + t{u})] = 0 \quad \mbox{for } 0 \le t < \delta.$$ 
\end{itemize} 
Under hypotheses (1)-(4), it holds that $A_0(\Delta) = 0$ for every $\Delta \in {\cal R}_0$. 
\end{theorem} 
 
\begin{proof} This is Theorem 2 in \cite{HC}, whose proof appears in Appendix B therein. 
\end{proof} 
 
\noindent The result also applies to more general non-relativistic {\em affine} spacetimes, in which case a fifth hypothesis is necessary. We omit this hypothesis here since it is automatically satisfied in Minkowski spacetime. Nevertheless, there exist relativistic cases where the fifth hypothesis cannot be imposed \cite{HC}.\\ 
 
\begin{corollary}\label{Fcor} 
No relativistic spatial localization observable according to Definition \ref{REMM0} satisfies condition (4)  
for elements $A^v(\Delta)$, with given $v\in \sS$, 
if the selfadjoint generator of time translations through $U$ is bounded below. 
\end{corollary} 
 
\begin{proof} 
Relativistic spatial localization observables, when (fixing $v\in \sS$ if any and) restricting to the bounded elements of ${\cal R}_0$ of any fixed foliation of parallel spacelike $3$-planes $\Sigma$, satisfy conditions (1)-(2) according to Definition \ref{REMM0}. Hence the theorem applies if (3) is valid. As a consequence, a relativistic spatial localization observable cannot satisfy (4) in view of the normalization condition of the POVM defined by $A^v$ on $\Sigma$, which would be violated by the thesis of Theorem \ref{HC}. 
\end{proof}

The above corollary applies in particular to the relativistic spatial localization observables of Definition \ref{REMM0} constructed for relativistic particles \cite{M23,C23,DM24,CDRM}, and also to the causal POL of \cite{Castrigiano2} when it is not constructed in terms of PVMs: for all these localization observables, condition (4) must be violated.

Halvorson-Clifton's theorem seems to impose a severe limitation on any notion of spatial localization for relativistic quantum systems. This is because the first three hypotheses, which seem genuinely necessary for defining a notion of spatial localization, conflict with the fourth, which, on the other hand, seems to express an elementary locality requirement in the quantum framework.  
If one could prove that hypothesis (4) cannot be relaxed because doing so would irreparably compromise fundamental locality requirements, the conclusion would be that no notion of spatial localization makes sense.  
 
A problematic point in this investigation is that commutativity of observables localized in causally separated regions is simply {\em assumed} as the basic description of locality in local quantum physics, in particular in the formalization due to Araki and Haag-Kastler  \cite{Haag,Araki} (see the next section). A deeper scrutiny, {\em independent} of the Araki-Haag-Kastler approach, seems to be necessary at this stage.
In particular, one should work in a framework where no requirement concerning the structure ($*$, $C^*$, or von Neumann algebra) is imposed on the family of observables localized in a spacetime region, and no {\em a priori} requirements are imposed on families of observables associated with distinct spacetime regions.
This type of analysis was initiated by Busch \cite{Buschloc} and developed by other authors, relying on the general theory of quantum measurement and on general local causality requirements concerning consecutive measurement processes. This analysis will be summarized in Section \ref{SECB}.

 \subsection{Relation with the Araki-Haag-Kastler framework}\label{AHK} In the {\em Araki-Haag-Kastler (AHK) formalism} \cite{Haag,Araki}, the von Neumann algebra $\gA$ of physically relevant operators of a quantum physical system described in the Hilbert space $\cH$ -- in particular the operators defining observables -- is generated by {\em local von Neumann algebras} $\gA({\cal O})$.  
There is a local von Neumann algebra $\gA({\cal O})$ for every open bounded set ${\cal O} \subset \bM$. 
Each such algebra contains operations and observables that are somehow physically associated with ${\cal O}$: the corresponding physical operations and measurements are performed therein. The identity operator $I$ is obviously common to all local algebras. 
If ${\cal O}$ is open but not bounded, $\gA({\cal O})$ is the von Neumann algebra generated by the family of bounded open subsets of ${\cal O}$. {\em Isotony} holds: 
$\gA({\cal O})\subset \gA({\cal O}_1)$ if ${\cal O}\subset {\cal O}_1$, and one of the fundamental assumptions is that operators belonging to algebras associated with causally separated regions must commute: 
$$[A_1,A_2]=0\quad \mbox{if $A_1\in \gA({\cal O}_1)$,  $A_2\in \gA({\cal O}_2)$ and ${\cal O}_1 \cap (J^+({\cal O}_2) \cup J^-({\cal O}_2)) = \emptyset$.}$$ 
The net of algebras is assumed to admit a strongly-continuous unitary representation of the Abelian translation group $\sV$ satisfying the {\em spectral condition}: the joint spectrum of the selfadjoint generators must lie in $\overline{\sV_+}$. This representation is assumed to extend to a full (continuous) unitary representation of the orthochronous Poincar\'e group. Finally, the Hilbert space contains a preferred Poincar\'e-invariant state, represented by a unit vector $\Omega$, the {\em vacuum vector}, which is cyclic: the subspace spanned by the vectors $A\Omega$ with $A\in \gA$ is dense in $\cH$.
    
There are a number of further requirements and fundamental results arising from the AHK approach. We shall, however, omit them all and stick to the elementary postulates above, since they are enough for our discussion. The reader may consult \cite{Haag,Araki} for classic general treatments of the subject.

Most features of this approach can be generalized to the case in which $\gA$ and every $\gA({\cal O})$ are {\em unital $*$-algebras} of operators on a given Hilbert space with a common invariant domain.

As already suggested, it seems natural to postulate that the self-adjoint operators $A^v(\Delta)$ and $A^{v'}(\Delta')$ belong to the local algebras associated with arbitrarily small open spacetime neighborhoods of $\Delta$ and $\Delta'$, respectively. If these open sets are causally separated,   
adopting the AHK formulation,  
then $[A^v(\Delta), A^{v'}(\Delta')]=0$. If $\Delta$ and $\Delta'$ lie in a common $3$-space $\Sigma$ and the distance between them is strictly positive, it is easy to prove the precise version of commutativity stated in requirement (4) of Theorem \ref{HC} by means of this argument.  
We shall see in the next sections that the hypothesis that $A^v(\Delta)$\footnote{Actually the POVM $\{A^v(\Delta), I-A^v(\Delta)\}$.} 
is localized in an arbitrarily small spacetime neighborhood of $\Delta$ is, however, disputable when dealing with relativistic spatial localization observables of {\em particles}: if these structures exist and can be constructed within the AHK framework, 
the localization region of the operators $A^v(\Delta)$ cannot be restricted to an arbitrarily small neighborhood of $\Delta$.

As is well known (see \cite{tesi} in particular), this negative result also follows from a fundamental theorem in local QFT known as the {\em Reeh--Schlieder theorem} \cite{Araki,Haag},
if one interprets $A^v(\Delta)$ as the mathematical representation of a detector localized in $\Delta$. If one assumes that this operator is positive and that its expectation value in the vacuum state is zero (no particles are present in that state!), then one faces a mathematical contradiction. For this reason, in the rigorous local formulation of scattering theory in the {\em Haag--Ruelle approach}, detectors are not assumed to be local observables but only {\em quasi-local observables} \cite{Araki,Haag}.

The AHK formulation makes it possible, in particular, to exploit {\em modular theory} \cite{Araki,Haag} (in the one-particle space of a scalar free quantum field) when the net of observables is formulated in terms of von Neumann algebras. As mentioned in the Introduction, this perspective has more recently inspired the lattice-theoretic approach to localization developed by Lechner and de Oliveira \cite{LdO}, which is based on the lattice of real projections rather than on the space of effects, as in \cite{CDRM}.

We discuss the AHK formulation further in a second paper \cite{M26}. Here we only stress that our analysis does not rely on this formulation, since we do not assume the AHK approach directly in this paper.

\section{No-signaling, relativistic consistency, and commutativity}\label{SECB} 
 
An important analysis of the validity of the commutativity hypothesis (4) in the Halvorson--Clifton theorem -- relying solely on general principles of measurement theory, and thus independently of the AHK framework \cite{Haag,Araki} -- was presented by Busch in \cite{Buschloc} (drawing on a previous technical paper \cite{BS} co-authored with Singh), prior to the Halvorson--Clifton paper. In \cite{HC}, Halvorson and Clifton justified their hypothesis (4) on the basis of Busch's analysis. 

The next three sections present and examine Busch's argument and some subsequent developments on the subject, in particular those due to S. Gudder and collaborators \cite{GN, AGG}. An updated and thorough discussion, including further relevant results and important remarks, appears in Section 2.2 of Beck's comprehensive PhD thesis \cite{tesi}. In Section 2.6 of Beck's analysis, it is explicitly stated that local commutativity does not seem to be an entirely well-motivated requirement. We shall see that our results are consistent with this view.
 
\subsection{Localization of (generalized) observables according to Busch} 
Busch-Singh's paper \cite{BS} and, more precisely, 
Busch's work \cite{Buschloc} focus on a version of Einstein causality called {\em weak Einstein causality}, also known as the {\em no-signaling condition} (NSC). 
A related causality condition mentioned by Busch is what we shall call the {\em relativistic consistency condition} (RCC), following Beck \cite{tesi}. 
To state them formally, we have to define the physical notion of localization for a generalized observable, following \cite{Buschloc}: ``[...] {\em we need to assume a physically 
meaningful association of observables with (bounded open) spacetime regions, in 
the sense that such observables can be measured by means of operations carried 
out within these regions}.''  
 
We stress that no further requirement concerning the structure ($*$, $C^*$, or von Neumann algebra) of the family of observables localized in a region is imposed, in contrast with the AHK approach.
Similarly, no {\em a priori} requirements are imposed on families of observables associated with distinct regions of spacetime.

In this context, observables are POVMs, with PVMs as a special case. A selfadjoint operator is not directly an observable; the associated observable is its PVM. 
An effect $E$ of a POVM, in spite of being selfadjoint, is not necessarily interpreted as an observable in its own right. Indeed, the associated observable would be the PVM of $E$, which, however, plays no physical role in the general case. Nevertheless, $E$ defines an associated elementary observable, the POVM $\{E, I-E\}$. In the case of an orthogonal projector $P$, viewed as a {\em test} on the quantum system, its PVM $\{P,I-P\}$ coincides with the associated elementary observable. 
 
Returning to the notion of a localized observable, we propose a slightly more general definition than Busch's (by relaxing the boundedness requirement), which is also technically more precise. It allows us to state NSC and RCC and to use them for both relativistic spatial localization observables and unsharp localization systems.\\ 
 
\begin{definition}\label{LOC}  
{\em A quantum observable $A$ (a POVM or a PVM) is said to be {\bf localized} in an open (not necessarily bounded) spacetime region ${\cal O}$ if it can be measured by means of operations carried 
out within that region.}\hfill $\blacksquare$\\ 
\end{definition}

\noindent This definition has to be understood as a {\em physical} definition. From the mathematical viewpoint, we are simply defining a preferred relation in the Cartesian product of the set of spacetime regions and the set of observables.

When specializing to localization observables, we also explicitly adopt an apparently obvious physical principle\footnote{Though the principle appears to be self-evident, the existence of entanglement phenomena such as those used in the original EPR argument could require a discussion.}, which we call the {\bf local detectability principle} ({\bf LDP}).  
The principle could already be stated for unsharp localization systems but, again, for future convenience we want to include the case in which the spatial localization region $\Delta$ is {\em unbounded}, as may occur for relativistic spatial localization observables.  
\begin{itemize} 
    \item[({\bf LDP})]  Let $A(\Delta)$ be an effect either of an unsharp localization system or of a spatial localization observable. If the elementary POVM $\{A(\Delta), I-A(\Delta)\}$ is localized in an open  region  ${\cal O}$, then  
    ${\cal O}\supset \Delta$. \hfill $\blacksquare$ 
\end{itemize} 
 
\noindent Let us illustrate the physical content of LDP. First, detections are interactions with instruments, namely detectors. The elementary POVM used above, $\{A(\Delta), I-A(\Delta)\}$, physically represents an observable associated with a detector with two outcomes for a given state $\rho$. 
\begin{itemize} 
\item{1}: the detector clicks, with probability $tr(A(\Delta)\rho)$,  
    \item{0}: the detector does not click, with probability $tr((I-A(\Delta))\rho)$. 
    \end{itemize} 
Since physical operations are not instantaneous and instruments occupy a region larger than the one in which the particle is detected, the detector operates in a spacetime region which, for instance, can be thought of as having the form ${\cal O} \equiv  (-\epsilon,\epsilon) \times {\cal D}_\Delta$, for some $\epsilon>0$, where  
$ {\cal D}_\Delta \subset \Sigma$ and the temporal coordinate refers to Minkowskian coordinates $(t,p)$ comoving with $n \perp \Sigma$. 
More generally,  the temporal coordinate $t$ may vary in different intervals $(-\epsilon_p,\epsilon'_p)$
where $p\in {\cal D}_\Delta$ and the maps ${\cal D}_\Delta \ni p \mapsto \epsilon_p, \epsilon'_p \in (0,+\infty)$ are at least continuous.
We can also consider different forms of the region ${\cal O}$ by including the worldlines of cables or instruments used to transmit information about where the particle is found. But it seems physically impossible not to assume that any region ${\cal O}$ where 
the elementary POVM $\{A(\Delta), I-A(\Delta)\}$ is localized includes a region like the one represented above and hence contains $\Delta$. 
This is what  LDP states. In the case where $\Delta$ is unbounded, we can assume the use of an infinite family of detectors placed on corresponding bounded sets of a partition of $\Delta$. Saying that the system is detected in $\Delta$ means that one detector of the family clicks.

\subsection{No-signaling condition and relativistic consistency condition} 
 
We now move on to discuss an elementary version of the {\em no-signaling principle} and a related principle that we shall call {\em relativistic consistency}. These physical principles are ways of making the locality principle precise in relativistic quantum theories \cite{Buschloc} without assuming {\em a priori} that locality corresponds to commutativity, as is supposed in the AHK framework. 
We need some technical definitions in order to state these principles. We shall deal with {\em finite} ($N<+\infty$) discrete POVMs because they are sufficient for the goals of this paper. However, all results could be extended to the case of infinite ($N=+\infty$) discrete POVMs.\\ 
 
\begin{definition}\label{DEFPOV} 
{\em A {\bf finite discrete} POVM $T$ is determined by $N<+\infty$ effects $\{T_1, \ldots, T_N\}$ on a Hilbert space $\cH$, so that $T_j\in \gB(\cH)$, $0<T_j\le I$, and $\sum_{j=1}^N T_j = I$. The numbers $j$, up to renaming operations, are the {\bf outcomes} of the measurement of $T$, the {\bf probability of measuring $j$ in the state $\rho$} being $tr(T_j\rho)$.} \hfill $\blacksquare$\\ 
\end{definition} 
 
\noindent Under quite general hypotheses of quantum measurement theory \cite{Buschbook,Moretti2}, the theory of the post-measurement state is encapsulated in the following definition.\\ 
 
\begin{definition} \label{DEFSEL} {\em The {\bf (selective) post-measurement state} $\rho_j^T\in \sS(\cH)$ of an initial state $\rho\in \sS(\cH)$ corresponding to outcome $j$ of a finite discrete POVM $T$ as in Definition \ref{DEFPOV}, when the probability $tr(\rho T_j)$ of outcome $j$ does not vanish, is 
\beq \rho^T_j = \sum_{k\in I_j} \frac{K^T_{jk}\rho K^{T\dagger}_{jk}}{tr(\rho T_j)}, \quad \text{with} \quad 
T_j =\sum_{k\in I_j} {K^{T\dagger}_{jk}K^T_{jk}} \label{EK}\:.\eeq 
The operators $K^T_{jk}\in \gB(\cH)$ are the {\bf Kraus operators} associated with outcome $j$ of the chosen measurement procedure (instrument) implementing the POVM $T$.   
The set $I_j$ is here assumed to be finite\footnote{It is possible to extend this formalism by relaxing these hypotheses \cite{Buschbook} using suitable weak topologies; however, our setting is sufficient for the goals of this section.}.\\ 
The post-measurement state of a {\bf non-selective measurement} of $T$, given a pre-measurement state $\rho \in \gB_1(\cH)$, is (where $\mathrm{tr}(\rho T_j)  \rho^T_{j} :=0$ if $\mathrm{tr}(\rho T_j)=0$)
\beq 
\rho^T :=  \sum_{j=1}^N \mathrm{tr}(\rho T_j)  \rho^T_{j} = \sum_{j= 1}^N \: \sum_{k\in I_j} {K^T_{jk}\rho K^{T\dagger}_{jk}} \label{out} 
\eeq 
The map that associates $\rho$ with $\rho^T_j$ or $\rho^T$ is a {\bf measurement (procedure)} of $T$. It is said to be {\bf efficient} if each $I_j$ consists of a unique element. \hfill $\blacksquare$\\} 
\end{definition}

\noindent As is well known, the Kraus operators $K^T_{jk}$ are not determined by the effects $T_j$: different instruments may implement the same POVM, and even a fixed instrument may admit different Kraus representations.

According to (\ref{out}), a measurement of $T$ is non-selective if all systems in an initially prepared ensemble of identical systems in the state $\rho\in \sS(\cH)$ 
are collected together after the measurement, regardless of the outcomes, yielding a new post-measurement state $\rho^T$.

We can now state our locality principles. 
 
The {\bf no-signaling condition} ({\bf NSC}), taking Definition \ref{LOC} into account, is stated as follows. 
\begin{itemize}  
\item[({\bf NSC})] Consider two finite discrete POVMs on a Hilbert space $\cH$: $T$ -- with effects $\{T_1, \ldots, T_N\}$ -- and $S$ -- with effects $\{S_1, \ldots, S_M\}$ -- {\em localized}, respectively, in ${\cal O}_T, {\cal O}_S\subset \bM$. Suppose that these observables are measured on a state $\rho\in \sS(\cH)$. Let $\rho^T\in \sS(\cH)$ denote the post-measurement state of a {\em non-selective} measurement of $T$. \\ 
The following identity must hold  
$$\mathrm{tr}(\rho S_k)= \mathrm{tr}(\rho^T S_k)\quad \mbox{for every state $\rho$ and every $k=1,\ldots, M$}$$ 
if ${\cal O}_T$ and ${\cal O}_S$ are {\em causally separated}. \hfill $\blacksquare$ 
\end{itemize} 
 
\noindent In other words, the outcome statistics of $S$ must be identical whether we use $\rho^T$ or the original state $\rho$.   
Violation of NSC would allow one to infer, by observing only outcomes of $S$ in ${\cal O}_S$, whether a measurement of $T$ occurred in the causally separated region ${\cal O}_T$. 
 
A related physical condition, which refers to {\em selective} measurements, is called the {\bf relativistic consistency condition} ({\bf RCC}): 
 
\begin{itemize}  
\item[({\bf RCC})] Consider two finite discrete POVMs on a Hilbert space $\cH$: $T$ -- with effects $\{T_1, \ldots, T_N\}$ -- and $S$ -- with effects $\{S_1, \ldots, S_M\}$ -- {\em localized}, respectively, in ${\cal O}_T, {\cal O}_S\subset \bM$. Suppose that these observables are measured on a state $\rho$. Let $\rho^T_j$ denote the post-measurement state of a measurement of $T$ with outcome $j$ and $\rho^S_i$ denote the post-measurement state of a measurement of $S$ with outcome $i$. \\ The following identity must hold 
$$\mathrm{tr}(\rho^T_j S_i) \mathrm{tr}(\rho^T_j)= \mathrm{tr}(\rho^S_i T_j) \mathrm{tr}(\rho^S_i)\quad \mbox{for every  $\rho\in \sS(\cH)$ and $i=1,\ldots, M$, $j=1,\ldots,N$}$$ 
if ${\cal O}_T$ and ${\cal O}_S$ are {\em causally separated}. \hfill $\blacksquare$ 
\end{itemize} 
This postulate requires invariance of the statistics of two consecutive {\em selective} measurements of the observables under interchange 
 of their order. In fact, $\mathrm{tr}(\rho^T_j S_i) \mathrm{tr}(\rho^T_j) $ is the probability of measuring first $j$ for $T$ and next $i$ for $S$, whereas $\mathrm{tr}(\rho^S_i T_j) \mathrm{tr}(\rho^S_i)$ is the probability of measuring first $i$ for $S$ and next $j$ for $T$, in both cases the measured state being $\rho$. The temporal order of two measurements located in causally separated regions has no physical meaning, and thus the statistics of these measurements should be invariant under interchange of order.

\subsection{NSC and RCC for L\"uders and more sophisticated types of measurements} 
 
To continue the analysis initiated by Busch and Singh \cite{BS} and Busch \cite{Buschloc} (see also Section 2.2 of Beck's thesis \cite{tesi} for developments and details) on the interplay among commutativity of effects, NSC, and RCC, we need some further general notions from quantum measurement theory. 
A special case of efficient measurement, which generalizes (finite discrete) {\em projective measurements}, is a {\bf L\"uders measurement}.\\ 
 
\begin{definition}\label{LUD} 
 {\em Consider a finite discrete POVM $T$ with $N<+\infty$ effects $\{T_1, \ldots, T_N\}$ on a Hilbert space $\cH$. A {\bf L\"uders measurement} of $T$ 
 is defined by the requirement $K^T_j:= \sqrt{T_j}$. In this case, we use the notation $\rho^T_{L,j }$ and $\rho^T_{L}$ for the selective and non-selective post-measurement states, respectively, of an initial state $\rho\in \sS(\cH)$. \hfill $\blacksquare$}\\ 
\end{definition} 
 
\begin{remark}\label{EMREM} 
 {\em In the general case, directly from the {\em polar decomposition theorem} (see, e.g., \cite{Moretti1}), any (Kraus) operator $K\in \gB(\cH)$ satisfying $K^\dagger K = T$ (e.g., associated with an efficient measurement of the POVM $\{T,I-T\}$) can be written as $K =V \sqrt{T}$, where $V$ is a partial isometry with initial space $\mathrm{Ker}(V)^\perp =\mathrm{Ker}( \sqrt{T})^\perp = \mathrm{Ker}(T)^\perp= \overline{\mathrm{Ran} (\sqrt{T})}= \overline{\mathrm{Ran} (T)}\:.$ 
If $\mathrm{Ker}(T)$ is trivial, i.e., $\mathrm{Ran}(T)$ is dense, then $V$ is necessarily an isometry. However, even if it is an isometry, in general $V$ is not unitary unless $K$ is bijective.} \hfill $\blacksquare$\\ 
\end{remark} 
 
We can now present the two fundamental results showing that, in physical terms, locality formalized by NSC and RCC implies commutativity of the effects of pairs of POVMs localized in causally separated regions. 
 
Before stating these results, a comment is necessary. As is clear from Beck's thorough analysis \cite{tesi} -- see Fig. 3 of \cite{tesi} in particular -- the two physical formulations of the locality principle, NSC and RCC, are not directly related to commutativity of these effects. A more direct relation holds between the {\em effects} of one observable and the {\em Kraus operators} of the other. In turn, this second type of commutativity implies commutativity of the effects. A more direct relation between effects and Kraus operators exists for L\"uders measurements or similar measurement procedures. 
 
A first general result due to Busch and Singh, which extends a previous result by L\"uders \cite{LUD} and connects NSC to commutativity, reads as follows:\\ 
 
\begin{proposition}\label{PROP}   
Consider a finite discrete POVM $T$ with effects $\{T_1,\ldots, T_N\}$ on $\cH$, and let $S$ be another effect on $\cH$. Then the equivalence 
$$[T_j, S]=0 \quad \forall j \quad \Leftrightarrow \quad \mathrm{tr}\left( \rho S\right) = \mathrm{tr}\left(\rho_L^T S\right)  \quad \mbox{for every state $\rho\in \sS(\cH)$,}$$ 
where the measurement of $T$ is supposed to be non-selective and of L\"uders type, holds in at least the following cases: 
\begin{itemize} 
\item[(i)] $S$ has a discrete spectrum with eigenvalues that can be ordered in decreasing order, $T$ arbitrary; 
\item[(ii)] $N=2$, so that $T_2 = I-T_1$, and $S$ is arbitrary. 
\end{itemize} 
\end{proposition} 
 
\begin{proof} This is Proposition 1 in \cite{Buschloc}, which was established in \cite{BS}.  \end{proof}  
A similar result can be proved in the RCC setting, as announced in \cite{Buschloc} without proof. However, an even stronger version holds, as established in \cite{GN} (see the discussion in \cite{tesi}).\\

\begin{proposition}\label{PROP2} 
Consider two finite discrete POVMs on $\cH$, $T$ with effects $\{T_1,\ldots, T_N\}$ and $S$ with effects $\{S_1,\ldots, S_M\}$, 
measured with selective L\"uders measurements. The following equivalence holds: 
$$ [T_j, S_i]=0\quad \forall i,j   \quad \Leftrightarrow \quad\mathrm{tr}\left(\rho^T_{L,j} S_i\right) \mathrm{tr}\left(\rho^T_{L,j} \right)   = \mathrm{tr}\left(\rho^S_{L,i} T_j\right)\mathrm{tr}\left(\rho^S_{L,i}\right)  \quad \forall i,j, \rho. $$ 
More generally, the implication $\Leftarrow$ is valid for efficient measurement procedures where $K^T_j = K^{T\dagger}_j$ and $K^S_i = K^{S\dagger}_i$ for all $i,j$. 
\end{proposition} 
 
\begin{proof} Suppose that the identity on the left-hand side of the double implication is true. For L\"uders measurements, $[T_j, S_i]=0$ (for all $i,j$) is equivalent to commutativity of the corresponding Kraus operators of one observable with the effects of the other, as immediately follows from the functional calculus of bounded selfadjoint operators (see, e.g., \cite{Moretti2}). Since these Kraus operators are in particular selfadjoint, the right-hand side of the double implication must be true as well, simply by applying the relevant definitions.  
The statement on the right-hand side of the double implication, for generic efficient measurements, is equivalent to $K^{T\dagger}_jS_iK^{T}_j=  K^{S\dagger}_iT_j K^{S}_i$ (for every $i,j$). If the Kraus operators are also selfadjoint (as, in particular, for L\"uders measurements), we can apply Theorem 2.10 in \cite{tesi} (arising from \cite{GN}), concluding that the identity on the left-hand side is true as well. 
\end{proof} 
When $T$ and $S$ are a pair of POVMs localized in causally separated regions, RCC and the above theorem imply that the effects must commute 
if the measurement procedures are efficient and the Kraus operators are selfadjoint, in particular for L\"uders measurement procedures.  
Similarly, if $S$ is itself an effect and the POVM $T$ and $\{S,I-S\}$ 
are localized in causally separated regions, NSC and Proposition \ref{PROP} imply that the effects must commute if the measurements are of L\"uders type.  
This substantial physical equivalence of no-signaling, consistency, and commutativity was already known for projective measurements \cite{BB}.\\

\begin{remark} \label{REMm} {\em 
The above results apply in particular to PVMs and orthogonal projectors corresponding to self-adjoint observables; this point already arises from L\"uders's original theorem \cite{LUD}. A notable consequence is that one can {\em prove} that quantum field operators $\hat{\phi}[f]$ represented in a Hilbert space $\cH$ and smeared by smooth compactly supported functions $f$ must commute when the supports of the test functions are causally separated -- instead of postulating this as in the AHK formulation -- if one assumes the no-signaling condition (or the relativistic consistency condition) together with other natural technical hypotheses.   
Specifically, it holds \beq [\hat{\phi}[f], \hat{\phi}[f']]=0 \quad \mbox{if $supp(f) \subset {\cal O}$ and $supp(f') \subset {\cal O}'$, with ${\cal O}, {\cal O}'$ causally separated,}\label{COM}\eeq 
where the identity holds on a dense common invariant domain $\gD \subset \cH$ of the $*$-algebra representation of the fields.   
The above identity is valid when NSC holds and: \begin{itemize} 
\item[(a)] $\hat{\phi}[f], \hat{\phi}[f']$ are essentially self-adjoint operators in the field Hilbert space, \item[(b)] they are assumed to be measured by projective measurements, \item[(c)]  
the selfadjoint observables $F(\overline{\hat{\phi}[f]}), F'(\overline{\hat{\phi}[f']})$ (i.e., their PVMs) are localized in ${\cal O}$ and ${\cal O}'$, respectively, for every choice of real Borel functions $F: \bR \to \bR$ and $F':\bR\to \bR$.  
\end{itemize} 
Indeed, using NSC for finite discrete PVMs constructed out of projectors associated with any finite partition of the spectra of $\overline{\hat{\phi}[f]}, \overline{\hat{\phi}[f']}$ (this is equivalent to choosing $F,F'$ as {\em simple functions}), Proposition \ref{PROP} -- or Proposition \ref{PROP2} when assuming RCC -- ensures commutativity of all projectors of the PVMs of $\overline{\hat{\phi}[f]}$ and $\overline{\hat{\phi}[f']}$ (see, e.g., \cite{Moretti1,Moretti2}). In turn, this implies $[e^{ia\overline{\hat{\phi}[f]}}, e^{ib \overline{\hat{\phi}[f']}}]=0$, and finally (\ref{COM}) follows on the invariant space $\gD$ via Stone's theorem.  
$\hfill \blacksquare$} 
\end{remark}

A final equivalence result, which extends Proposition \ref{PROP} to {\em generic non-efficient} measurements of finite discrete POVMs and effects, is provided by Beck's Theorem 3.12 \cite{tesi}, but it uses a further, rather strong hypothesis.\\ 
\begin{theorem}[Beck]\label{PROP3}   
Consider a finite discrete POVM $T$ with effects $\{T_1,\ldots, T_N\}$ on $\cH$, and let $S$ be another effect on $\cH$. Then the following equivalence holds: 
$$[K^T_{jk}, S]= [K^{T\dagger}_{jk}, S] = 0 \quad \forall j, k \quad \Leftrightarrow \quad \mathrm{tr}\left( \rho S\right) = \mathrm{tr}\left(\rho^T S\right) \quad \mbox{and}\quad \mathrm{tr}\left( \rho S^2\right) = \mathrm{tr}\left(\rho^T S^2\right)  \quad \forall \rho\:, $$ 
where a generic non-selective measurement of $T$ according to (\ref{out}) is considered. \\ \end{theorem}

\begin{remark}  \label{Remc} {\em $[K^T_{jk}, S]= [K^{T\dagger}_{jk}, S] = 0$ immediately implies $[T_j, S]=0$ from the latter identity in (\ref{EK}). \hfill $\blacksquare$}\end{remark}

\section{Commutativity requirement and relativistic locality} 
We are now in a position to study the relationship between Propositions \ref{PROP}, \ref{PROP2} and the validity of the commutativity hypothesis (4) in Halvorson-Clifton's theorem for relativistic spatial localization observables, including the comparison of effects with different tensorial indices $v$. 
First, we prove that, in principle, these results -- in which NSC/RCC and L\"uders measurements play a crucial role -- imply commutativity for the effects of an unsharp localization system associated with causally separated regions, provided some apparently natural hypotheses are fulfilled.  
Next, we prove that the causal-separation hypotheses are never satisfied, simply because of the physical nature of a particle and its propensity to localize in a unique place. Hence the commutativity assumption does {\em not} arise by this route. In this case, failure of commutativity does {\em not} conflict with locality principles such as NSC and RCC. 
 
\subsection{NSC, RCC and commutativity of effects} 
Busch's crucial result about condition (4) in the Halvorson-Clifton theorem can now be stated precisely, also in a more general version that takes Proposition \ref{PROP2} into account.\\ 
 
\begin{theorem}\label{B}   
Consider two regions $\Delta, \Delta' \in {\cal R}$ for either a relativistic spatial localization observable or an unsharp localization system $(\cH, {\cal R}, A, U)$ (in this case the specifications $v,v'$ below are not necessary), 
and assume there exist open regions of spacetime ${\cal O}_\Delta \supset \Delta$ and ${\cal O}_{\Delta'} \supset \Delta'$ such that  
\begin{itemize}  
\item[(i)] the POVM $\{A^v(\Delta), I-A^v(\Delta)\}$ is localized in ${\cal O}_\Delta$; 
\item[(ii)] the POVM $\{A^{v'}(\Delta'), I-A^{v'}(\Delta')\}$ is localized in ${\cal O}_{\Delta'}$. 
\end{itemize} 
If, moreover, 
\begin{itemize} 
\item[(a)] measurements of these POVMs are performed using L\"uders procedures,  
\item[(b)] ${\cal O}_\Delta$ and ${\cal O}_{\Delta'}$ are causally separated, 
\item[(c)] NSC or RCC holds (in the latter case the measurement procedures in (a) can be weakened to efficient measurements with selfadjoint Kraus operators), 
\end{itemize} 
then 
$[A^v(\Delta), A^{v'}(\Delta')]=0.$ 
\end{theorem} 
 
\begin{proof}  
This immediately follows from case (ii) of Proposition \ref{PROP}, with the POVM $T$ described by $\{A^v(\Delta), I-A^v(\Delta)\}$ and the effect $S$ given by $A^{v'}(\Delta')$. Alternatively, it immediately follows from Proposition \ref{PROP2}, with $S$ interpreted as the POVM $\{A^{v'}(\Delta'), I-A^{v'}(\Delta')\}$. 
\end{proof}

\begin{remark}\label{REMBECK} {\em An alternative result can be proved by taking advantage of Theorem \ref{PROP3} and Remark \ref{Remc}. 
Compared with Theorem \ref{B}, this approach permits much more general measurement procedures, which may be non-efficient and defined in terms of non-selfadjoint operators. As in Theorem \ref{B}, hypotheses (i), (ii), (a), and (b) should be assumed, as well as (c) in the case of NSC. However, to exploit the implication $\Leftarrow$ of Theorem \ref{PROP3} together with Remark \ref{Remc}, one should assume, for every $\rho\in \sS(\cH)$ (where we omit the possible specifications $v,v'$), 
 $$\mathrm{tr}\left( \rho A(\Delta')^2\right) = \mathrm{tr}\left(\rho^T A(\Delta')^2\right)\quad \mbox{with $T$ given by $\{A(\Delta), I-A(\Delta)\}$}$$ 
which, in principle, might again be justified as a consequence of NSC. 
However, in the identity above, $A(\Delta')^2$ may not be viewed as an elementary effect embedded in the elementary POVM $\{A(\Delta'), I-A(\Delta')\}$, but rather as a selfadjoint observable in its own right (with its own PVM). In particular, according to Definition \ref{LOC}, we should also assume that {\em there is a further physical procedure to measure the PVM of the selfadjoint operator $A(\Delta')^2$} 
and that this PVM is localized in a region causally separated from ${\cal O_{\Delta}}$. 
In this case, using the standard elementary formulation of quantum theory (see, e.g., \cite{Moretti2}), $\mathrm{tr}\left( \rho A(\Delta')^2\right)$ is interpreted as the expectation value $\langle A(\Delta')^2 \rangle_\rho$ of the selfadjoint observable $A(\Delta')^2$ in the state $\rho$. 
In the AHK framework, adopting this additional hypothesis also seems fairly natural, but the use of the AHK setting is exactly what we are trying to avoid. An alternative possibility is to introduce another elementary POVM $\{A(\Delta')^2, I-A(\Delta')^2\}$ and assume that it is localized in a region causally separate from ${\cal O}_\Delta$.  
From a general perspective, the fact that the mathematical expression of the no-signaling condition may hold for an effect $S$: 
$\mathrm{tr}\left( \rho S\right) = \mathrm{tr}\left(\rho^TS\right)$ for every $\rho\in \sS(\cH)$, 
but not for $S^2$: 
$\mathrm{tr}\left( \rho S^2\right) \neq  \mathrm{tr}\left(\rho^T S^2\right)$ for some $\rho\in \sS(\cH)$, is a concrete possibility. A counterexample due to Heinosaari and Wolf \cite{HW} is discussed by Beck in Section 2.5 of \cite{tesi}.} \hfill $\blacksquare$\\ 
\end{remark} 
 
In view of Theorem \ref{B}, it is then straightforward to see that requirement (4) in Halvorson-Clifton's theorem \ref{HC} holds under some further reasonable physical assumptions: 
\begin{itemize} 
\item[(i)] all measurements of the POVMs $\{A(\Delta), I-A(\Delta)\}$, for $\Delta \in {\cal R}$, are performed using the L\"uders procedure assuming NSC -- or with efficient selfadjoint Kraus operators and assuming RCC; 
\item[(ii)] for any $\Delta \in {\cal R}$, $\{A(\Delta), I-A(\Delta)\}$ can be localized in open regions ${\cal O}^\epsilon_\Delta\supset  \Delta$ that shrink around $\Delta$ as $\epsilon\to 0^+$: ${\cal O}^{\epsilon}_\Delta$ is the causal completion of a covering of $\Delta$ by open balls of radius $\epsilon>0$ within the $3$-plane $\Sigma \ni \Delta$. 
\end{itemize} 
If $\Delta, \Delta'\subset \Sigma$ belong to ${\cal R}$ and their distance within $\Sigma$ is $\epsilon>0$, then ${\cal O}^{\epsilon/4}_\Delta$ and ${\cal O}^{\epsilon/4}_{\Delta'}$ are clearly causally separated, and their  
intersections with $\Sigma$ 
have distance at least $\epsilon/2$.   
If ${a}\in \sV_+$, then $\Delta$ and $\Delta'+ t {a}$ remain causally separated for sufficiently small $|t|<\delta$.   
Thus, Busch-Singh's Proposition \ref{PROP} implies that $A(\Delta)$ and $A(\Delta'+ t{a})$ commute, confirming that condition (4) in Halvorson-Clifton's theorem \ref{HC} is satisfied. 
 
\subsection{Commutativity and Relativistic Spatial Localization Observables}\label{NCn} 
We want to study the validity of hypothesis (4) in the case where the set of effects $A^v(\Delta)$ is organized into, or extends to, families of POVMs forming a {\em relativistic spatial localization observable} according to Definition \ref{REMM0}. 
However, this structure may be considerably weakened by using an extension of the notion of unsharp localization system where (1) unbounded sets $\Delta$ are allowed and (2) $\Sigma \setminus \Delta\in {\cal R}$ if $\Delta\in {\cal R}$ and $\Delta \subset \Sigma$. Yet, requirements (c) and (d) in Definition \ref{REMM0} do not play any role here.\\

\begin{proposition} \label{PROP22} Let $(\cH, {\cal R}, A, U)$ be a relativistic spatial localization observable, let $\Sigma$ be a rest space of a reference frame, and suppose that LDP holds. If, for $\Delta\in {\cal R}$ with $\Delta \subset \Sigma$, the elementary POVM $\{A^v(\Delta), I-A^v(\Delta)\}$ is localized in the open set ${\cal O} \subset \bM$, then ${\cal O} \supset \Sigma$ necessarily. 
\end{proposition}

\begin{proof}  Using the definition of POVM, $\{A^v(\Delta), I-A^v(\Delta)\}= \{I-A^v(\Delta^c), A^v(\Delta^c)\}$ where $\Delta^c:= \Sigma \setminus \Delta\in {\cal R}$ due to the very definition of relativistic spatial localization observable. At this point, since $\{A^v(\Delta^c), I-A^v(\Delta^c)\}$ is localized in ${\cal O}$, LDP implies that $\Delta^c \subset {\cal O}$. In summary, ${\cal O}\supset \Delta \cup \Delta^c = \Sigma$. 
\end{proof}

\begin{corollary} \label{CORR22} Let $(\cH, {\cal R}, A, U)$ be a relativistic spatial localization observable, suppose that LDP holds, and consider $\Delta, \Delta' \in {\cal R}$. If the corresponding elementary POVMs 
$\{A^v(\Delta), I-A^v(\Delta)\}$ and $\{A^{v'}(\Delta'), I-A^{v'}(\Delta')\}$ are localized, respectively, in ${\cal O}_{\Delta}$ and ${\cal O}_{\Delta'}$, then  
these sets cannot be causally separated (even if $\Delta$ and $\Delta'$ are causally separated). 
\end{corollary} 
 
\begin{proof} As ${\cal O}_{\Delta}\supset \Sigma$ and the latter is a {\em spacelike Cauchy surface}, $J^+({\cal O}_{\Delta'})\cup J^-({\cal O}_{\Delta'})$ must meet $\Sigma$ -- and thus ${\cal O}_{\Delta}$ -- somewhere.  
\end{proof} 
 
Consequently, $\{A^v(\Delta), I-A^v(\Delta)\}$ and  
$\{A^{v'}(\Delta'), I-A^{v'}(\Delta')\}$ do not satisfy the basic hypothesis (b) of Theorem \ref{B}. We conclude that commutativity, i.e., the natural generalization of hypothesis (4) in Halvorson-Clifton's theorem, is not physically mandatory for {\em relativistic spatial localization observables} and for {\em unsharp localization systems} that are extendable to the former type of structure. 
This result is actually quite general. Even establishing a more general version of Theorem \ref{B}, also taking Remark \ref{REMBECK} into account and considering measurement procedures much more general than L\"uders(-like) ones, would not remove the obstacle to deriving commutativity from locality. Corollary \ref{CORR22} shows that, if a relativistic spatial localization observable exists, the region where elementary observables are localized is the whole physical space at the corresponding instant of time, even if the observables refer to bounded detection regions. Two localization regions, unlike the detection regions, are {\em never} causally separated, and this precludes the use of locality arguments relying on NSC or RCC.  
The only apparent alternative is to abandon from the outset the physically natural principle LDP, which played a crucial role in the proof of Proposition \ref{PROP22}.

In summary, what we have established with Proposition \ref{PROP22} and Corollary \ref{CORR22} has a clear physical meaning. We are considering systems that localize in 
a unique region, namely at a {\em single point} of space in ideal detection procedures, and not in {\em several} regions {\em simultaneously}: in other words, we are considering a {\em particle}. A particle has the propensity to localize in a {\em single} place, unlike extended systems\footnote{Since localization arises from  a physical interaction, it is nevertheless questionable whether localization in a region smaller than the particle's Compton wavelength is physically meaningful, in view of the possibility of pair creation.}. 
We are simply exploring the consequences in measurement theory of this fundamental propensity characteristic of particles. Using a POVM on $\Sigma$ of a relativistic spatial localization observable means that we have filled the entire rest space $\Sigma$ of an inertial reference frame with infinitely many detectors. We turn them on at the instant $t$ -- which defines $\Sigma$ -- and turn them off immediately afterwards. At that instant, the particle must necessarily be detected in a unique detector. If we focus our attention on an arbitrarily chosen region $\Delta\subset \Sigma$, the system may or may not be found there at time $t$. In the first case, however, we are sure that it is simultaneously detected nowhere in $\Delta^c$; in the second case, we are sure that it is simultaneously detected somewhere in $\Delta^c$. By observing what happens in $\Delta$, we therefore have information about what is happening in the whole of $\Sigma$ at the instant $t$. In this sense, it seems natural that  
$\{A^v(\Delta), I-A^v(\Delta)\}$  
should be thought of as localized {\em everywhere} in $\Sigma$, and there is no contradiction with locality because, as explained above, we are using the propensity of the quantum system under consideration, namely a particle, to localize somewhere and in a unique place if its position is measured.  
If we admit that the system is allowed to localize ``simultaneously'' in several places, i.e., in causally separated regions, we are no longer dealing with the notion of a particle, and the notion of a relativistic spatial localization observable is not appropriate for such a system.
Whether or not particles exist in the above sense is simply a matter for experimental investigation. In this work, we assume that they do exist. In \cite{M26}, we make the stronger assumption that elementary particles in Wigner's sense are of this sort, at least for massive scalar bosons.

This analysis can also be applied  
to the case where the regions $\Delta$ and $\Delta'$ belong to {\em different} rest spaces of {\em different} reference frames, provided these regions are causally separated. This follows from the fact that at least some relativistic spatial localization observables can be extended to the family of generic smooth spacelike Cauchy surfaces \cite{DM24}. Furthermore, for sufficiently regular spacelike-separated subsets $\Delta$ and $\Delta'$, there is a spacelike Cauchy surface $S$ that contains both \cite{BS06}. We can use $S$ as our notion of rest space in which to detect the particle, and we can recast the previous argument, this time assuming that this generalized rest space $S$ is filled with detectors. The result can be generalized further by taking $\Delta$ and $\Delta'$ of minimal regularity (Lipschitz achronal Borel sets) and exploiting a maximal achronal set $S$ that includes both \cite{CDRM}. 
Finally, even in the situation where we switch on {\em only two} detectors in causally separated regions $\Delta$ and $\Delta'$ of a given rest space $\Sigma$, the above reasoning 
{\em could}  
hold in a weaker version. If the particle is found in $\Delta$, it cannot be found in $\Delta'$ simply in view of its propensity to localize in a unique place. Unlike the previous case, however, there is now the possibility that the particle is localized in neither of the switched-on detectors, since they do not exhaust all possibilities. However, this situation is more delicate to analyze, and we do not address it here. 
 
As a consequence of the discussion above, from an AHK perspective, we expect that every $A^v(\Delta)$ of a relativistic spatial localization observable, if any, does not belong to a local algebra of observables 
$\gA({\cal O})$ 
with ${\cal O}$ an arbitrarily small neighborhood of $\Delta$. Otherwise, commutativity should hold for sufficiently causally separated $\Delta$ and $\Delta'$. This ${\cal O}$ has to be taken as a neighborhood of the whole $\Sigma$ and cannot be restricted further.

Finally, it is important to emphasize that we are not claiming that hypothesis (4) of the Halvorson--Clifton theorem \ref{HC} cannot hold in general. Rather, we argue that, from a physical standpoint, commutativity of effects {\em of localization observables}, if any, is not a consequence of NSC or RCC (through Propositions \ref{PROP},\ref{PROP2}, and \ref{B} in particular). Indeed, there exist in the literature relativistic spatial localization observables for a particle-antiparticle fermion system \cite{Castrigiano2,DRCM26}, described by families of PVMs, such that, for every $\Delta,\Delta'\subset \Sigma$, one has 
$[A(\Delta),A(\Delta')] = 0$ -- and more precisely $A(\Delta)A(\Delta')= A(\Delta')A(\Delta)=A(\Delta\cap \Delta')= 0$ if $\Delta\cap \Delta' =\emptyset$. {\em This is true on account of the very definition of PVM and not as a consequence of Busch's theorem}. For these families of PVMs, hypothesis (4) also holds as a consequence of standard properties of orthogonal projectors and CC (see a proof in Section \ref{SECFC}).  
For these relativistic spatial localization observables concerning a fermion \cite{Castrigiano2}, however, the conclusion of Halvorson-Clifton's Theorem \ref{HC} does not hold, because in general $A(\Delta) \neq 0$. The reason is that the energy-boundedness requirement (3) in Theorem \ref{HC} is not satisfied. Furthermore, upon restriction to the particle space, namely the positive-energy sector, commutativity ceases to hold, and the observables are therefore described by POVMs.

Nevertheless, if one insists on maintaining the energy-boundedness requirement together with additivity and translation covariance as well, then commutativity cannot hold simply because of the Halvorson-Clifton result. However, for a quantum particle, understood as above, this fact does not seem to pose a problem.

\section{Possible commutativity recovery via local conditional POVMs}\label{secFC} 
In this section, we propose a different viewpoint on the localization of quantum systems, in which measurements are explicitly performed in finite regions of space, namely laboratories. By construction, the observables we are going to define should in principle be causally independent when associated with causally separated spatial regions. 
We focus on the probabilities of detecting the system in one of the detectors of a laboratory, given that it is already known that the system is found in the laboratory. 
Hence the physical interpretation of the procedure we are about to introduce involves {\em conditional probabilities}. 
 
\subsection{Localization in laboratories} Realistic localization experiments in $\bM$ are performed in {\em laboratories}, which occupy finite spatial regions. We shall also assume that they are {\em causally complete} in their spacetime description: everything that may happen in such a spacetime region should be determined by physical actions performed therein, at least at a macroscopic level.  
Let $\Sigma$ be a flat $3$-dimensional plane. We define a {\bf laboratory} $L(\Delta_0)\subset \bM$, with {\bf space} given by a bounded set $\Delta_0\in \cB(\Sigma)$, as the causal completion\footnote{We warn the reader that this orthogonality relation is not the one used in \cite{CDRM}, where {\em achronal separation} was instead considered to describe the orthomodular lattice of the so-called {\em causal logic} of Minkowski spacetime.} defined in (\ref{CComp}):
\beq L(\Delta_0) := (\Delta_0^\pperp)^\pperp\:. \label{LAB}\eeq 
In particular, if $\Delta_0$ is open in the relative topology, the resulting open set $L(\Delta_0)$, equipped with the Lorentzian manifold structure induced by $\bM$, is a globally hyperbolic spacetime in its own right, with $\Delta_0$ as a possible smooth spacelike Cauchy surface. $L(\Delta_0)$ also contains other {\em curved} smooth spacelike Cauchy surfaces, though we shall focus only on the bounded flat 3-space $\Delta_0$ of the laboratory $L(\Delta_0)$. 
 
Given a laboratory $L(\Delta_0)$, we want to localize a quantum system in the subregions $\Delta \subset \Delta_0$ of its flat bounded portion of rest space $\Delta_0$, even if the system can generally be detected in the whole spacetime. For the moment, we once more assume that the entire space is filled with detectors represented by a relativistic spatial localization observable $A$. In Section \ref{CFS} we shall try to get rid of this unrealistic perspective. The idea is to define a POVM referring only to the laboratory, starting from a notion of localization $A$ defined on the whole space $\Sigma$. This ``restricted'' POVM should describe in which subregion $\Delta \in \cB(\Delta_0)$ of the laboratory the system is detected {\em if the system is detected in $\Delta_0$}. Notice that the POVM we are looking for should be normalized on $\Delta_0$ and not on the whole $\Sigma$. It concerns the fraction of events in $\Delta$ over the events in $\Delta_0$ (and not over those in the whole $\Sigma$). 
If we denote the effects of this POVM by $B_{\Delta_0}(\Delta)$, then  
$tr(\rho B_{\Delta_0}(\Delta))$ is {\em the probability of finding the system in $\Delta$ if it is found in $\Delta_0$}. In particular, $tr(\rho B_{\Delta_0}(\Delta_0))=1$ because it is the probability of finding the system in $\Delta_0$ if it is found in $\Delta_0$! 
 
\subsection{Localization in a laboratory from a relativistic spatial localization observable}\label{secfc1}

Suppose we have at our disposal a full relativistic spatial localization observable denoted by $A^v: {\cal R}\to \gB(\cH)$ or just a POVM $A$ on a specific spacelike $3$-plane $\Sigma$.  
Take $\Delta_0\in \cB(\Sigma)$. We intend to re-normalize $A$ on this $\Delta_0$.   
The idea is to consider families of operators of the form 
\beq B_{\Delta_0}(\Delta)= \frac{1}{\sqrt{A(\Delta_0)}}A(\Delta)\frac{1}{\sqrt{A(\Delta_0)}}\quad \mbox{where $\Delta \in \cB(\Delta_0)$.}\label{NEwW} \eeq 
This is because, if these operators can be constructed, they define a normalized POVM on $\Delta_0$. However, some technical issues arise immediately in connection with the proposed definition, in particular the existence of the inverse operators written there. Let us examine the less obvious ingredients of the formula above. 
By functional calculus (and selfadjointness), it is easy to prove that $Ker(\sqrt{A(\Delta_0)}) = Ker(A(\Delta_0))$ simply because $\gB(\cH)\ni A(\Delta_0)\geq 0$. Therefore, if $Ker(A(\Delta_0))=\{0\}$,  
then a  
selfadjoint operator $A(\Delta_0)^{-1/2} : Ran(\sqrt{A(\Delta_0)})\to \cH$ is well defined. 
Notice that the domain is dense because 
$\overline{Ran(\sqrt{A(\Delta_0)})}= Ker(\sqrt{A(\Delta_0)})^\perp = \cH$. 
At this point, an important technical result concerning the condition $Ker(A(\Delta_0))=\{0\}$ follows from even stronger results due to Castrigiano\footnote{I am grateful to Castrigiano for some clarifications about this issue.} \cite{Castrigiano2}, which extend previous results by Hegerfeldt. 
In the next proposition, the whole structure of a relativistic spatial localization observable $(\cH, A,{\cal R}, U)$ is not necessary. \\

\begin{proposition}\label{TEOCASTR} 
Let $(\cH, A,{\cal R}, U)$ be a weakened version of a scalar relativistic spatial localization observable as in Definition \ref{REMM0}, with the following modifications: 
${\cal R}$ in (a) consists of Borel regions in the rest spaces orthogonal to a preferred Minkowskian reference frame $u \in \sT_+$, and covariance as in (d) is required only for the subgroups of spatial and temporal translations in $u$.  
Assume the following. 
\begin{itemize} 
\item[(1)] Referring to the above Minkowskian reference frame $u$, the square-mass operator $$M^2:= P_0^2 - \sum_{k=1}^3 P_k^2$$ of the strongly-continuous unitary representation $U$ of $\widetilde{ISO(1,3)_+}$ (defined on the {\em G{\aa}rding} domain\footnote{It is essentially selfadjoint on this domain, as is every central element of the universal enveloping algebra of the Lie algebra of $\widetilde{ISO(1,3)_+}$ represented in terms of generators of $U$. This condition does not depend on the chosen Minkowskian reference frame.}) is non-negative. 
\item[(2)] Energy is positive: the selfadjoint generator $H^u= -P_0 = P^0$ of time evolution along $u$ through $U$ (i.e., a past-directed temporal translation along $u$) has non-negative spectrum. 
\item[(3)] CT is satisfied for the rest spaces $\Sigma_u$ normal to $u$. 
\end{itemize} 
Let $\Delta_0\subset \Sigma_u$ be an element of ${\cal R}$. Then, 
$$Ker(A(\Delta_0))\neq \{0\} \quad \Rightarrow 
\quad \overline{\Sigma_u \setminus \Delta_0} = \Sigma_u\:.$$ 
If $(\cH, A,{\cal R}, U)$ is a relativistic spatial localization observable that satisfies CC and conditions (1) and (2) (which do not depend on the reference frame), then the conclusion is true for every rest space $\Sigma$ and every $A^v(\Delta)$ with $\Delta \subset \Sigma$. 
\end{proposition} 
 
\begin{proof}  
The condition $Ker(A(\Delta_0))\neq \{0\}$ is equivalent to $A(\Delta_0)\psi =0$ for some $\psi \in \cH\setminus\{0\}$, which, in turn, is equivalent to  
$A(\Delta_0^c)\psi =\psi$, where $\Delta_0^c:= \Sigma_u \setminus \Delta_0$.  
The condition $A(\Delta_0^c)\psi =\psi$ for $\psi \neq 0$ implies, in fact, that $\Delta_0^c$ is {\em essentially dense} in $\Sigma_u$: $\overline{\Delta^c_0\setminus N} = \Sigma_u$  
for every Lebesgue-null set $N\subset \Sigma_u$ as a consequence of Theorem (8) in \cite{Castrigiano2}, also taking into account the comment following Definition (5) in the same reference. Using $N=\emptyset$, we obtain $\overline{\Sigma_u \setminus \Delta_0}= \Sigma_u$, which proves the claim. The last statement is an immediate consequence of what has been proved and of the definition of a relativistic spatial localization observable. 
\end{proof}

\begin{remark} {\em This type of result is consistent with the idea that a relativistic quantum system cannot be sharply localized -- with probability $1$ -- in a non-dense spatial region, in particular a bounded one: the probability of finding it arbitrarily far away in space is always positive, even if it tends to vanish, when time evolution is causal (and energy and the square-mass operator are positive, as expected for {\em relativistic systems}). An important issue is whether, conversely, localized probability distributions can be approximated arbitrarily well by admissible states of this type. A general discussion of these issues, as well as concrete examples for single fermions, appears in \cite{Castrigiano2}. For single bosons, some results are presented in \cite{M23,C23}.} \hfill $\blacksquare$\\ 
\end{remark} 
 
\noindent We are now in a position to give a rigorous technical meaning to the right-hand side of (\ref{NEwW}), even in a slightly more general setting. \\ 
 
\begin{lemma}\label{LEMMAB} Let $A: \cB(\Sigma)\to \gB(\cH)$ be a POVM, where $\Sigma\subset \bM$ is a spacelike $3$-plane. Choose  
\begin{itemize} 
\item[(1)] $\Delta_0 \in \cB(\Sigma)$ such that $Ker(A(\Delta_0))=  \{0\}$\footnote{This is in particular true if $\Delta_0$ is bounded with $Int(\Delta_0)\neq \emptyset$, and (1),(2),(3) of Proposition \ref{TEOCASTR} hold.}; 
\item[(2)] an isometry $V: \cH \to \cH$ whose final space is $Ran(V) (= \overline{Ran(V)})=:\cH_V$. 
\end{itemize} 
For every $\Delta\in \cB(\Delta_0)$ 
there is a unique operator $B^V_{\Delta_0}(\Delta)\in \gB(\cH_V)$ such that  
\beq \label{BD1} 
\left\langle  \frac{1}{\sqrt{A(\Delta_0)}} V^\dagger\psi \left|A(\Delta)\frac{1}{\sqrt{A(\Delta_0)}}V^\dagger\phi\right.\right\rangle = \langle \psi| B^V_{\Delta_0}(\Delta) \phi\rangle \quad \mbox{if $\psi, \phi \in V(Ran (\sqrt{A(\Delta_0)}))$.} 
\eeq 
In addition, $B^V_{\Delta_0}(\Delta)$ is the unique bounded extension to $\cH_V$ of either of the equivalent compositions under the overlines in 
\beq B^V_{\Delta_0}(\Delta)=\overline{\left(\frac{1}{\sqrt{A(\Delta_0)}}V^\dagger\right)^\dagger A(\Delta)\frac{1}{\sqrt{A(\Delta_0)}}V^\dagger}=\overline{V\frac{1}{\sqrt{A(\Delta_0)}}A(\Delta)\frac{1}{\sqrt{A(\Delta_0)}}V^\dagger}\:. \label{BD1ext}\eeq 
%To specify its natural domain, set $A_0:=A(\Delta_0)$ and define 
%$$C:D(C)\to\cH,\qquad C:=\frac{1}{\sqrt{A_0}}V^\dagger,\qquad D(C):=V(Ran(\sqrt{A_0}))\subset\cH_V.$$ 
%Then $C^\dagger:D(C^\dagger)\to\cH_V$, where $D(C^\dagger)=Ran(\sqrt{A_0})$ and $C^\dagger=V\frac{1}{\sqrt{A_0}}$. The natural domain of the composition is 
%$$D(C^\dagger A(\Delta)C)=\{\phi\in D(C)\:|\:A(\Delta)C\phi\in D(C^\dagger)\}=D(C)=V(Ran(\sqrt{A_0})).$$ 
%Thus, the bar in (\ref{BD1ext}) denotes the unique bounded extension of $C^\dagger A(\Delta)C$ from this domain to $\cH_V$. 
Furthermore, $B^V_{\Delta_0}(\Delta)$ is an effect with $B^V_{\Delta_0}(\Delta_0)= VV^\dagger$, the orthogonal projector onto the final space of $V$. 
\end{lemma} 
 
\begin{proof} Use the notation introduced in the statement and set $A_0:=A(\Delta_0)$ and $A_\Delta:=A(\Delta)$. 
Since $Ker(A_0)=\{0\}$, $Ran(\sqrt{A_0})$ is dense in $\cH$. Moreover, $A_0^{-1/2}$ is selfadjoint with $D(A_0^{-1/2})=Ran(\sqrt{A_0})$, and $V^\dagger:\cH_V\to\cH$ is unitary. Therefore, $C$ is closed and densely defined. Since $\Delta\subset\Delta_0$, we have 
$0\leq A_\Delta\leq A_0\leq I$ and 
consequently, for every $\psi\in D(C)$, 
$$0\leq \langle C\psi|A_\Delta C\psi\rangle\leq \langle C\psi|A_0C\psi\rangle=||V^\dagger\psi||^2=||\psi||^2.$$ 
The Cauchy--Schwarz inequality for the positive sesquilinear form defined by $A_\Delta$ thus gives 
$$|\langle C\psi|A_\Delta C\phi\rangle|\leq ||\psi||\:||\phi||\quad \mbox{for every $\psi,\phi\in D(C)$.}$$ 
For a fixed $\phi\in D(C)$, the map 
$$D(C)\ni\psi\mapsto\langle C\psi|A_\Delta C\phi\rangle$$ 
is therefore bounded with respect to the norm of $\cH_V$. By the definition of the adjoint operator, this means that $A_\Delta C\phi\in D(C^\dagger)$ and 
$$\langle C\psi|A_\Delta C\phi\rangle=\langle\psi|C^\dagger A_\Delta C\phi\rangle\quad \mbox{for every $\psi,\phi\in D(C)$.}$$ 
Since this holds for every $\phi\in D(C)$, $D(C^\dagger A_\Delta C)=D(C)$. The densely defined operator $C^\dagger A_\Delta C$ is thus bounded on $D(C)$, with norm not larger than $1$, and admits a unique bounded extension to $\cH_V$. Denote this extension by $B^V_{\Delta_0}(\Delta)$. This proves (\ref{BD1}), and the quadratic estimate above implies, by continuity, that $0\leq B^V_{\Delta_0}(\Delta)\leq I$. Since $V^\dagger:\cH_V\to\cH$ is unitary, we also have 
$$D(C^\dagger)=Ran(\sqrt{A_0}),\qquad C^\dagger=V\frac{1}{\sqrt{A_0}},$$ 
where the equality of operators includes equality of their domains. Therefore, on $D(C)$, 
$$C^\dagger A_\Delta C=\left(\frac{1}{\sqrt{A_0}}V^\dagger\right)^\dagger A_\Delta\frac{1}{\sqrt{A_0}}V^\dagger=V\frac{1}{\sqrt{A_0}}A_\Delta\frac{1}{\sqrt{A_0}}V^\dagger.$$ 
Taking the unique bounded extensions to $\cH_V$ proves both equalities in (\ref{BD1ext}). Finally, if $\Delta=\Delta_0$, then 
$$\langle C\psi|A_0C\phi\rangle=\langle V^\dagger\psi|V^\dagger\phi\rangle=\langle\psi|\phi\rangle\quad \mbox{for every $\psi,\phi\in D(C)$},$$ 
and therefore 
 \beq 0\leq B^V_{\Delta_0}(\Delta) \leq I\quad \mbox{and} \quad B^V_{\Delta_0}(\Delta_0)=VV^\dagger. \label{BD2}\eeq 
\end{proof}

\begin{remark} {\em The proof shows, without assuming a strictly positive lower bound for $A(\Delta_0)$, that $B^V_{\Delta_0}(\Delta)$ is the unique bounded extension of 
$$V\frac{1}{\sqrt{A(\Delta_0)}}A(\Delta)\frac{1}{\sqrt{A(\Delta_0)}}V^\dagger$$ 
from $V(Ran(\sqrt{A(\Delta_0)}))$ to $\cH_V$. If, in addition, $A(\Delta_0)\geq cI$ for some $c>0$, then $Ran(A(\Delta_0))=\cH$ and $\frac{1}{\sqrt{A(\Delta_0)}}\in\gB(\cH)$. In this situation the operator above is already bounded and everywhere defined on $\cH_V$, so that no extension is necessary and 
\beq B^V_{\Delta_0}(\Delta)=V  \frac{1}{\sqrt{A(\Delta_0)}}A(\Delta)\frac{1}{\sqrt{A(\Delta_0)}}V^\dagger \label{BD222} 
\eeq 
holds directly.} \hfill $\blacksquare$\\ 
\end{remark} 
 
\noindent The effects $B^V_{\Delta_0}(\Delta)$ are our candidates for defining a notion of conditional spatial localization in the rest space $\Delta_0$ of a laboratory.  
To avoid issues with normalization of the POVM we want to define, it is convenient to assume that $VV^\dagger=I$, which is equivalent to requiring that $V$ be unitary. 
The final result can now be stated and proved in the general case.\\

\begin{proposition}\label{PROPB} 
Let $A$ be a relativistic spatial localization observable and $V: \cH \to \cH$ a unitary operator.  
Assume that $\Delta_0\in \cB(\Sigma)$, with $\Sigma \subset \bM$ a spacelike $3$-plane, satisfies $0<|\Delta_0|<+\infty$ and $Ker(A^v(\Delta_0))=\{0\}$ for every $v\in \sS$. 
Then 
\begin{itemize} 
\item[(a)] $\cB(\Delta_0) \ni \Delta \mapsto B^{v,V}_{\Delta_0}(\Delta) \in \gB(\cH)$, with $B^{v,V}_{\Delta_0}(\Delta)$ defined in (\ref{BD1}) for $A^v$ in place of $A$, is also a normalized POVM and is absolutely continuous with respect to the Lebesgue measure on $\Sigma$; 
\item[(b)] If the unitary representation of the universal covering of the proper orthochronous Poincar\'e group used in the definition of $A$ is $U: \widetilde{\cal P}_+ \to \gB(\cH)$, the covariance relation holds 
$$U_g B^{v,V}_{\Delta_0}(\Delta) U_g^{-1} = B^{gv, U_gVU_g^{-1}}_{g\Delta_0}(g\Delta)\quad \mbox{for every $g\in  \widetilde{\cal P}_+$, $\Delta \in \cB(\Delta_0)$.} $$ 
\end{itemize} 
\end{proposition}

\begin{proof} (a) We omit the tensorial index $v$ for notational simplicity. In view of (\ref{BD2}), the claim holds if, for every $\psi\in \cH$, the map  
$$\cB(\Delta_0) \ni \Delta \mapsto \langle \psi | B^V_{\Delta_0}(\Delta)\psi \rangle =: \mu_\psi(\Delta)\in [0,+\infty)$$ 
is a finite positive Borel measure absolutely continuous with respect to the Lebesgue measure. It remains to prove that the map  
is $\sigma$-additive for every given $\psi\in \cH$ and satisfies the stated absolute continuity property. This is true if $\psi$ belongs to the dense subspace $V(Ran(\sqrt{A(\Delta_0)}))$ of $\cH$, by construction and in view of (\ref{BD1}). In particular, every positive finite measure $\mu_\psi(\Delta):=\langle \psi | B^V_{\Delta_0}(\Delta)\psi \rangle$ is absolutely continuous with respect to the Lebesgue measure for $\psi \in V(Ran(\sqrt{A(\Delta_0)}))$, and if $\psi \in \cH_V$, there is a sequence $V(Ran(\sqrt{A(\Delta_0)}))\ni \psi_n \to \psi$, so that $\mu_{\psi_n}(\Delta)\to \mu_\psi(\Delta)$ for every $\Delta \in \cB(\Delta_0)$ by (\ref{BD1}). 
Since the Lebesgue measure of $\Delta_0$ is finite, the {\em Vitali-Hahn-Saks} theorem \cite{Yosida} implies that $\mu_\psi$ is $\sigma$-additive -- so that it is a finite positive Borel measure on $\Delta_0$ -- and is also absolutely continuous with respect to the Lebesgue measure since each $\mu_{\psi_n}$ has that property.\\ 
(b) This follows directly from (d) of Definition \ref{REMM0}, equation (\ref{BD1}), and the spectral functional calculus for the selfadjoint operator  
$A^v(\Delta_0)$, using the fact that (d) of Definition \ref{REMM0} implies: $U f(A^v(\Delta)) U^{-1} = f(UA^v(\Delta)U^{-1})$ for every Borel measurable $f: \bR \to \bC$ and every unitary $U\in \gB(\cH)$ \cite{Moretti2}. 
\end{proof}

\subsection{$B_{\Delta_0}(\Delta)$ as conditional POVM: the role of the {\em gentle measurement lemma}}\label{secfc11} 
Suppose that the initial state is $\rho$ and that we are measuring a scalar relativistic spatial localization observable $A$. For simplicity, we assume that $A$ is scalar, since the argument below immediately generalizes.   
As said previously, we are interested in the probability of finding the system -- prepared in the state $\rho$ -- in the Borel subsets $\Delta \subset \Delta_0$ {\em when we know that the system  
is found in $\Delta_0$}. Here $\Delta_0\subset \Sigma$ has finite Lebesgue measure and, for this purpose, we may assume that $\Delta_0$ is bounded. We also explicitly suppose that $Ker(A(\Delta_0))=\{0\}$, in order to define the new POVM $B^V_{\Delta_0}: \cB(\Delta_0) \to \gB(\cH)$ upon a choice of the unitary $V$. 
 
To prove that $B^V_{\Delta_0}$ is a good candidate for representing this type of conditional probability, at least for some choice of $V$, we adopt a heuristic approach. If we know that the system is found in $\Delta_0$, assuming that $A$ was measured with a measurement procedure described by Kraus operators 
as in Remark \ref{EMREM}, we expect the state to have the form 
\beq {\rho}^V_{\Delta_0} :=\frac{V\sqrt{A(\Delta_0)}{\rho} \sqrt{A(\Delta_0)}V^\dagger}{tr(A(\Delta_0) {\rho})}\in\sS(\cH)\:,\label{rhoV}\eeq 
for some state ${{\rho}}\in \sS(\cH)$ and some partial isometry $V$, which we explicitly suppose to be unitary and, for the moment, arbitrary. Notice that $tr(A(\Delta_0) {\rho})\neq 0$ as a consequence of our hypothesis $Ker(A(\Delta_0))=\{0\}$ and ${\rho} \in \sS(\cH)$. 
At this point, a straightforward computation based on the cyclic property of the trace and (\ref{rhoV}) yields, for every $\rho \in \sS(\cH)$, 
\beq\label{EqP} 
\frac{tr(\rho A(\Delta))}{tr(\rho A(\Delta_0))} =  tr\left(\rho^V_{\Delta_0} B^V_{\Delta_0}(\Delta) \right) \:. \eeq 
We have therefore found that the probability of localizing the system in $\Delta\subset \Delta_0$, given that we know it was found in the laboratory rest space $\Delta_0\subset \Sigma$, is precisely described by the POVM $B^V_{\Delta_0}: \cB(\Delta_0) \to \gB(\cH)$.  
This heuristic reasoning, however, has two related shortcomings: we assumed that the initially prepared state was $\rho$ and not $\rho^V_{\Delta_0}$, which in the discussion above plays only an intermediate role. Furthermore, the above heuristic interpretation refers to {\em two measurements}: one that transforms ${\rho}$ into ${\rho}^V_{\Delta_0}$ and a subsequent one  
that detects the system in some $\Delta\subset \Delta_0$. Instead, we want to study a situation in which {\em only one measurement takes place on $\rho$: we only consider the events occurring in $\Delta_0$, disregarding those outside $\Delta_0$, and, without performing further measurements, we focus on the subset of such events occurring in the sets $\Delta\subset \Delta_0$}. 
To address these issues, taking advantage of a suitable version of the so-called {\em gentle measurement lemma}, we are about to show that we can in fact replace ${\rho}^V_{\Delta_0}$ 
with $\rho$ in (\ref{EqP}), provided the probability of finding the system in $\Delta_0$ for the initial state $\rho$ is close to $1$. This requires $V=I$. (We shall return to the case $V\neq I$ at the end of the section.)
We use an improved version of the result originally due to A. Winter \cite{W99}. The term ``gentle measurement lemma'' is more recent and now also refers to an even tighter estimate in the finite-dimensional case. Winter proved the crucial inequality used in the following theorem in the finite-dimensional case but, as is well known and as is clear from the proof itself \cite{W99}, finite dimensionality plays no essential role. \\ 
 
\begin{theorem}[Gentle measurement lemma]\label{GMT} 
Let $T \in \gB(\cH)$ be a positive operator on the Hilbert space $\cH$, and consider a state $\rho \in \sS(\cH)$ such that $tr(\rho T)\geq ||T||(1-\delta)>0$. The following bound holds: 
\beq 
\left|\left|\rho -\frac{\sqrt{T} \rho \sqrt{T}}{tr(T\rho)}\right|\right|_1 \leq 2 \sqrt{\delta} + \delta\:. 
\eeq 
\end{theorem} 
 
\begin{proof}  See the Appendix.
\end{proof} 
 
\noindent It is easy to see that the statement is generally false if one replaces $\sqrt{T}$ with $V\sqrt{A}$ for an effect $A$ and a unitary map $V$ which does not commute with $\rho$.  

As a consequence of the above result, we assume for now that $V=I$ and, in applying the theorem, we henceforth use the following notation for this special case:
\beq 
B^v_{\Delta_0}(\Delta) := B^{v,I}_{\Delta_0}(\Delta) \quad \mbox{and, for the scalar case} \quad B_{\Delta_0}(\Delta) := B^{I}_{\Delta_0}(\Delta)\:. 
\eeq 
 
\begin{proposition}\label{TAPP}  Consider a quantum system described in the Hilbert space $\cH$ and a relativistic spatial localization observable with effects $A^v: {\cal R} \to \gB(\cH)$. Let $\Delta_0 \in \cB(\Sigma)$, where $\Sigma\subset\bM$ is a spacelike $3$-plane, be such that  
$0<|\Delta_0|<+\infty$ and $Ker(A^v(\Delta_0))=\{0\}$\footnote{This is in particular true if $\Delta_0$ is bounded with $Int(\Delta_0)\neq \emptyset$ and (1),(2),(3) of Proposition \ref{TEOCASTR} hold.}. \\ If $\rho \in \sS(\cH)$ is such that the probability of finding the system in $\Delta_0$ satisfies $$tr(\rho A^v(\Delta_0)) \geq  1-\delta >0\:,$$ then \beq\left|tr(\rho B^v_{\Delta_0}(\Delta))  -  \frac{tr(\rho A^v(\Delta))}{tr(\rho A^v(\Delta_0))}\right| \leq 2 \sqrt{\delta} + \delta \quad \mbox{for every measurable $\Delta \subset \Delta_0$.} \label{3ee2}\eeq 
More generally,  
\beq|tr({\rho}^I_{\Delta_0} B)-  tr(\rho B)| \leq  (2 \sqrt{\delta} + \delta)  ||B|| \quad \mbox{for every $B\in \gB(\cH)$.} \label{3ee}\eeq 
\end{proposition}

\begin{proof} We shall write $A,B$ respectively in place of $A^v, B^v$ for brevity. 
Let us start with the latter inequality. We have  
$|tr({\rho}^I_{\Delta_0} B)-  tr(\rho  B)| = |tr(({\rho}^I_{\Delta_0} -\rho) B)| \leq || \rho^I_{\Delta_0}- \rho||_1  ||B||$. The gentle measurement lemma, applied with $||T||^{-1} T = A(\Delta_0)$, yields the second inequality in the statement. 
The first inequality then follows immediately from the second by taking $B= B_{\Delta_0}(\Delta)$, using identity (\ref{EqP}), and observing that $||B_{\Delta_0}(\Delta)|| \leq 1$. 
\end{proof} 
 
\noindent We conclude that, for states whose probability of finding the system in $\Delta_0$ is close to $1$, $B_{\Delta_0}(\Delta)$ measures the {\em  probability of finding the system in $\Delta$ if we know that it was found in $\Delta_0$.} Statistically speaking, 
\beq 
  tr\left({\rho} B_{\Delta_0}(\Delta) \right) \simeq  \frac{tr(\rho A(\Delta))}{tr(\rho A(\Delta_0))}\qquad \mbox{where $\Delta \subset \Delta_0 \quad \mbox{and} \quad\Delta,\Delta_0 \in {\cal R}$,}\label{fractio} 
\eeq 
is the ratio of the number of detections in $\Delta$ to the total number of detections {\em observed in the laboratory} $\Delta_0$ when an ensemble of identical systems is prepared in the initial state ${\rho}$. 
 
We finally consider  the general case of $B^{v,V}_{\Delta_0}(\Delta)$ with $V\neq I$. Suppose that $V:\cH \to \cH$ is unitary and, informally speaking,  
\beq B^{v,V}_{\Delta_0}(\Delta) = V \frac{1}{\sqrt{A^v(\Delta_0)}} A^v(\Delta)  \frac{1}{\sqrt{A^v(\Delta_0)}} V^\dagger\label{QQWQW}\eeq 
whose rigorous interpretation is (\ref{BD1}) if we only know that $Ker(A^v(\Delta_0))=\{0\}$ without strictly positive lower bounds for $A^v(\Delta_0)$. The spectral functional calculus permits us to rewrite the identity above equivalently as 
$$B^{v,V}_{\Delta_0}(\Delta) =  \frac{1}{\sqrt{V A^v(\Delta_0)V^\dagger}} V A^v(\Delta)V^\dagger  \frac{1}{\sqrt{VA^v(\Delta_0)V^\dagger}} = B'^{v}_{\Delta_0}(\Delta)$$ 
where, clearly, 
$$B'^{v}_{\Delta_0}(\Delta) = \frac{1}{\sqrt{A'^v(\Delta_0)}} A'^v(\Delta)  \frac{1}{\sqrt{A'^v(\Delta_0)}}\quad\mbox{with} \quad  A'^v(\Delta_0) := V A^v(\Delta_0) V^\dagger\:.$$ 
Now $(\cH, {\cal R}, A', U')$ is a (causal) relativistic spatial localization observable if $(\cH, {\cal R}, A, U)$ is one -- where $U'_g := VU_gV^\dagger$. 
Furthermore, the selfadjoint generators of time translations are bounded from below for the relativistic spatial localization observable $A'$ if and only if those for $A$ are bounded from below.
Hence the apparently more general  
POVM $B^{v,V}_{\Delta_0}(\Delta)$ is still a conditional POVM and simply reduces to a redefinition of the initially used relativistic spatial localization observable. In spite of this straightforward mathematical result, we stress that this redefinition may have physical relevance.

\subsection{Counterfactual interpretation of $B_{\Delta_0}(\Delta)$ and generalization}\label{CFS} 
In establishing the fundamental estimate (\ref{fractio}), we assumed that measurements are performed with a relativistic localization observable $A$ {\em whose detectors are also distributed outside $\Delta_0$}, on the whole rest space $\Sigma\supset \Delta_0$, and not only in $\Delta_0$.  
However, it seems physically plausible that  
$tr(\rho B_{\Delta_0}(\Delta))$ coincides with the said conditional probabilities, 
also in experiments where the detectors of $A$ {\em are switched off outside} $\Delta_0$ or are not even placed on $\Sigma\setminus \Delta_0$. 
This assumption is similar to other {\em counterfactual} interpretations used also in elementary physics\footnote{Think of the operational definition of the electric field $\vec{E}$ in terms of its action on a test charge which is not present.}. We stress that we are not saying that $A$ does not exist; we are merely saying that the corresponding detectors are placed only in $\Delta_0$, so that detection events are registered there only. 
In support of this counterfactual interpretation of $B_{\Delta_0}$, we stress that the fraction on the right-hand side of (\ref{fractio}) can be experimentally computed by referring only to the frequencies of events detected in $\Delta_0$ (even though, if $\Sigma$ is filled with detectors, there are also events {\em outside} $\Delta_0$).  
Furthermore, if two relativistic spatial localization observables $A$ and $A'$ coincide on $\Delta_0$, they give rise to the same fraction (\ref{fractio}) and the same trace $tr(B_{\Delta_0}(\Delta) \rho_{\Delta_0}^I)$ for any given state $\rho$.  
There is another point in support of this interpretation. As we proved in \cite{DM24}, localization observables can be extended to generic non-flat Cauchy surfaces $S$ (and in \cite{CDRM} this result has been further generalized to maximal achronal sets). However, if $\Delta_0$ is common to a pair of Cauchy surfaces $S,S'\supset \Delta_0$, the corresponding effects coincide: $A_S(\Delta_0) = A_{S'}(\Delta_0)$, and thus so do the corresponding fractions (\ref{fractio}). Furthermore, from the definitions of $B_{\Delta_0}(\Delta)$, this coherence property makes $tr(B_{\Delta_0}(\Delta) \rho)$ independent of $S$ and $S'$ when, obviously, also $\Delta \subset \Delta_0 \subset S\cap S'$.  
In summary, the above counterfactual interpretation can be rephrased by stating that  
the POVM $B_{\Delta_0}$ for $\Delta \subset \Delta_0$ measures the probability of detecting the system in $\Delta$ if we know that the system is found in the laboratory based on $\Delta_0$\footnote{This is at least reasonably valid for states with probability of finding the particle in $\Delta_0$ close to $1$, but this interpretation could be valid in general and could be experimentally tested.}. This is true independently of (a) the existence of other detectors outside the laboratory and (b) the choice of other instruments outside $\Delta_0$ (in its causal completion) that may give rise to a complete relativistic spatial localization throughout spacetime and agree with the instruments in $\Delta_0$.

If one accepts the viewpoint outlined above, we can even generalize the idea illustrated above for the POVM $B_{\Delta_0}$. It is not necessary to start from a complete relativistic spatial localization observable $(\cH, {\cal R}, A, U)$ in order to construct it. In particular, the normalization $A^v(\Sigma)=I$ is not required. It is sufficient to have an $IO(1,3)_+$-covariant family 
of positive operators (non-normalized POVMs) ${\cal R} \ni \Delta \mapsto T^v(\Delta)\in \gB(\cH)$ that is $\sigma$-additive on every $\Sigma$ and has the property that $Ker (T^v(\Delta_0)) =\{0\}$ if $\Delta_0$ is bounded with non-empty interior, typically the closure of a bounded open set included in a rest space $\Sigma$ that can be used as the space of a laboratory.  
It is not even necessary that these operators be available throughout spacetime. It is sufficient that ${\cal R}$ include only the Borel sets on planes $\Sigma$ that lie in a spacetime neighborhood of a given $\Delta_0$, which can be interpreted as the rest space of a laboratory. 
In this case, a normalized (and covariant) POVM localized in every such $\Delta_0$ can be defined by 
\beq \Delta_0 \supset \Delta \mapsto B^v_{\Delta_0}(\Delta):=\frac{1}{\sqrt{T^v(\Delta_0)}} T^v(\Delta) \frac{1}{\sqrt{T^v(\Delta_0)}}\:,\label{GENgen}\eeq 
whose rigorous meaning is the one given in (\ref{BD1}) if $0\in \sigma(T^v(\Delta_0))$. 
We stress in particular that (\ref{fractio}) is still valid if one refers to the {\em effects} $A^v(\Delta)= ||T^v(\Delta_0)||^{-1}T^v(\Delta)$ by direct application of the gentle measurement lemma. Here, the physical condition required to interpret $tr(\rho B^v_{\Delta_0}(\Delta))$ as a conditional localization probability is that $\frac{tr(\rho T^v(\Delta_0))}{||T^v(\Delta_0)||} \geq  1-\delta$ with $0\leq \delta <\sp<1$.

 The appearance of unitary operators $V$ 
\beq \Delta_0 \supset \Delta \mapsto B^v_{\Delta_0}(\Delta):=V\frac{1}{\sqrt{T^v(\Delta_0)}} T^v(\Delta) \frac{1}{\sqrt{T^v(\Delta_0)}}V^\dagger\:,\label{GENgen2}\eeq 
does not change the structure, since (\ref{GENgen2}) can be rewritten as in (\ref{GENgen}) simply by referring to the non-normalized POVM 
${\cal R} \ni \Delta \mapsto T'^v(\Delta) := VT^v(\Delta)V^\dagger \in \gB(\cH)$.
 
\subsection{Possible commutativity of local conditional effects} 
We restrict attention to the scalar case, without loss of generality. 
Consider a pair of laboratories determined by corresponding non-empty Borel sets $\Delta_0\subset \Sigma$ 
and $\Delta_0' \subset \Sigma'$, which are assumed to be open and bounded. We also suppose we have two  
conditional POVMs as above, $\cB(\Delta_0)\ni \Delta \mapsto  B_{\Delta_0}(\Delta) \in \gB(\cH)$ and $\cB(\Delta'_0)\ni \Delta' \mapsto  B_{\Delta'_0}(\Delta')\in \gB(\cH)$. 
It seems natural to assume that LDP is also valid for the POVMs under consideration. 
At this point, the same argument leading to Proposition \ref{PROP22}, {\em noting that the POVMs are now normalized on the respective finite sets $\Delta_0$ and $\Delta_0'$}, shows that, for $\Delta \subset \Delta_0$, if $\{ B_{\Delta_0}(\Delta), I- B_{\Delta_0}(\Delta)\}$ is localized in ${\cal O}_\Delta$, then ${\cal O}_\Delta\supset \Delta_0$, {\em instead of the much more demanding condition ${\cal O}_\Delta\supset \Sigma$ found when dealing with the effects $A(\Delta)$}. Similarly, if $\Delta' \subset \Delta'_0$ and $\{B_{\Delta'_0}(\Delta'), I- B_{\Delta'_0}(\Delta')\}$ is localized in ${\cal O}_{\Delta'}$, then ${\cal O}_{\Delta'}\supset \Delta'_0$. 
{\em In the absence of obstructions of a different nature which may arise from the specific structure of the considered operators}, if $\Delta_0$ 
and $\Delta_0'$ are sufficiently far apart and causally separated, it should be possible to choose the open sets ${\cal O}_\Delta$ and ${\cal O}_{\Delta'}$ causally separated. At this point,  
NSC or RCC would imply that $[B_{\Delta_0}(\Delta), B_{\Delta'_0}(\Delta')]=0$ through Theorem \ref{B}. With a more sophisticated discussion relying upon Remark \ref{REMBECK}, the same result could be established 
by exploiting Beck's Theorem \ref{PROP3}. 

We stress that this is not a proof of commutativity, since we did not prove that ${\cal O}_\Delta$ and ${\cal O}_{\Delta'}$ can be chosen causally separated! Different kinds of issues may arise.  
However, the problem leading to Corollary \ref{CORR22} does not occur here, because the (normalized) POVMs are now defined only on $\Delta_0$ and $\Delta'_0$  
and they measure probabilities of {\em conditional statements}. For instance, $B_{\Delta_0}(\Delta)$ measures the probability of finding the particle in $\Delta \subset \Delta_0$, conditioned on the fact that we already know that the particle was found in $\Delta_0$ and thus {\em not} in $\Delta'_0$. Hence it is not related to the analogous probabilities of conditional events in $\Delta'\subset \Delta'_0$ measured by $B_{\Delta'_0}(\Delta')$, which considers a disjoint family of events.

Another technical point in favor of possible commutativity of the effects  
$B_{\Delta_0}(\Delta)$ and $B_{\Delta'_0}(\Delta')$ (for $\Delta_0$ and $\Delta'_0$ causally separated) is the following observation. If we explicitly start from a relativistic spatial localization observable $A$, the family of effects $B_{\Delta_0}(\Delta)$, with $\Delta_0$ varying over the bounded Borel subsets of all spatial $3$-planes $\Sigma$ for which $Ker(A(\Delta_0))=\{0\}$ and $\Delta\in \cB(\Delta_0)$, does {\em not} satisfy the hypotheses of Halvorson-Clifton's Theorem \cite{HC}, even if the energy-boundedness requirement is satisfied. This is because the {\em additivity condition} 1 fails: additivity is not valid across different laboratories. A more complicated relation holds. If $\Delta_0 \cap \Delta_0'=\emptyset$ and $\Delta_0,\Delta'_0\subset \Sigma$, 
$$B_{\Delta_0 \cup \Delta_0'}(\Delta) \psi = \frac{1}{\sqrt{A(\Delta_0\cup\Delta_0' )}} \left(\sqrt{A(\Delta_0)} B_{\Delta_0}(\Delta \cap \Delta_0) \sqrt{A(\Delta_0)} \right. $$ 
$$ +  \left. \sqrt{A(\Delta'_0)} B_{\Delta'_0}(\Delta\cap \Delta'_0) \sqrt{A(\Delta'_0)}   \right)   \frac{1}{\sqrt{A(\Delta_0\cup\Delta_0' )}}\psi $$ 
for $\Delta \subset \Delta_0\cup \Delta_0'$ and for $\psi$ in the standard domain of the operator composition on the right-hand side. 
 
In a second paper \cite{M26}, we show that conditional POVMs describing spatial localization in laboratories exist and satisfy the commutativity requirement stated above. The construction is presented in the spirit of the rigorous Araki-Haag-Kastler formulation of (free) QFT in Minkowski spacetime for a massive real scalar field. 

We emphasize that, if $B_{\Delta_0}(\Delta)$ belongs to a local algebra in the sense of AHK, then $B_{\Delta_0}(\Delta)$ should arise as the restriction to the one-particle space of an operator defined on the full Hilbert space, typically a Fock space of a local quantum field theory on Minkowski spacetime with vacuum state $\Omega$. As already observed, the {\em Reeh--Schlieder theorem} then implies that positivity of $B_{\Delta_0}(\Delta)$ on the full Hilbert space is incompatible with the requirement that $\langle \Omega|B_{\Delta_0}(\Delta)\Omega \rangle$ vanish; in other words, a {\em dark-count} phenomenon occurs entirely because of the properties of the vacuum state. This is a well-known issue in local quantum physics, recently re-examined in a quantitative framework in \cite{FC26} by considering a general detector model and deriving a theoretical inequality that compares detector efficiency with the dark-count phenomenon. In that work, the authors also raise the question of whether the spacetime localization region of the detector represented by $B_{\Delta_0}(\Delta)$ is genuinely concentrated around $\Delta$ or should instead be understood as a larger region encompassing the whole laboratory, as suggested by our perspective.

In the scattering-theory literature based on local QFT, by contrast, a different standpoint is usually adopted: detector operators are assumed to have zero vacuum expectation value:
$
\langle \Omega |B\Omega\rangle = 0
$,
and operators $B$ are therefore {\em not} local (but \emph{quasi-local}, as in Haag--Ruelle scattering theory \cite{Haag,Araki}), in contrast to the operators described above.

\section{Conclusions and outlook}  
In this work, we studied the issue of commutativity in the representation of relativistic locality within the description of localization observables for relativistic quantum systems. The celebrated no-go theorem of Halvorson and Clifton shows that the commutativity of causally separated localization effects is incompatible with other apparently natural assumptions concerning the notion of spatial localization. This implies, in particular, that localization operators -- elements of suitable families of POVMs associated with the rest spaces of inertial observers -- cannot belong to local algebras of observables in the sense of Araki--Haag--Kastler. 
Adopting Busch's analysis of commutativity, this paper established, in particular, that the commutativity requirement is not, in this context, a necessary consequence of more elementary principles such as no-signaling or relativistic consistency. This holds for elementary systems, i.e., particles, as soon as one assumes that, in highly idealized experiments in which a position measurement is performed throughout the whole rest space of a reference frame at a given time, they have the propensity to localize in a unique place. In this sense, in the case under consideration, the failure of commutativity does not represent a problem for causality. 
  
On the other hand, we proved that, in principle, {\em both commutativity and localization} can be restored when one refers to a more sophisticated, and perhaps more realistic, notion of {\em conditional localization}. Here, one attempts to define POVMs normalized on a restricted region of space, namely a laboratory. These POVMs correspond to the conditional probability of finding a particle in a subregion of the laboratory, given that it has been detected there. A formal expression for this type of POVM, for states of a particle with small probability of being detected outside the laboratory, has been suggested in terms of general localization observables associated with full rest spaces. It is argued that effects belonging to pairs of these conditional POVMs, associated with causally separated laboratories, may in principle commute. A crucial technical result in this interpretation is the so-called {\em gentle measurement lemma}. 
  
It would be interesting to explore possible connections between these ideas and the modular-theoretic approach developed in \cite{LdO}. This issue will be investigated elsewhere.

As a final remark, we observe that, in light of the results presented here -- and also considering the failure to describe localization by means of a self-adjoint operator -- we may conclude that the notion of position for a physical object in a fully developed relativistic quantum theory has a rather peculiar status. In contrast with non-relativistic quantum theory, there does not appear to be a preferred notion of position observable.
Although one expects that, in the non-relativistic regime, all such notions converge to the naive one -- namely, the position operator (as shown, in particular, for the relativistic localization observables constructed in \cite{M23,DM24}) -- this is no longer the case in the fully relativistic setting.
Furthermore, any notion of spatial localization must necessarily be unsharp. Position in a given reference frame cannot be identified with the joint spectrum of a triple of self-adjoint operators. From this perspective, unsharp observables assume a foundational role and should not be regarded merely as approximate or effective versions of sharp ones. It is plausible that sharp observables, i.e., self-adjoint operators, are instead to be interpreted primarily as generators of symmetries.

A second paper \cite{M26} constructs explicit examples of general localization observables and conditional local position observables using the local AHK approach to QFT. In particular, it is shown that this type of localization observable can be constructed from the (smeared) normally ordered stress--energy operator, taking advantage of some known results on quantum energy inequalities. 

\section*{Acknowledgments} The author is grateful to D. Castrigiano and C. De Rosa for useful discussions about the issues presented in this work. The author is also grateful to the two anonymous referees for their constructive comments, which helped improve the exposition. This work has been written within the activities of INdAM-GNFM.

\section*{Declaration statements}
%{\bf Author contributions}: All authors equally contributed to produce this work. All authors read and approved the final manuscript.\\
{\bf Conflict of Interest}: The author declares no conflict of interest.\\
{\bf Ethical Statement}: This work does not involve human participants, animals, or sensitive data, and therefore no ethical approval was required.\\
{\bf Informed Consent}: Not applicable.\\
{\bf Data Availability}: No datasets were generated or analyzed in this study. All relevant information is contained within the article.\\
{\bf Funding}: This research received no external funding.

\section*{Appendix} 
\appendix

\section{Appendix: Validity of H-C hypothesis (4) for some fermionic PVMs}  \label{SECFC} According to CC, $A(\Delta)\leq A(J^+(\Delta) \cap (\Sigma+ t {a})) \quad \mbox{if $t>0$.}$ Since the effects are actually orthogonal projectors,  
this inequality is equivalent to (see, e.g., \cite{Moretti2})  $$A(\Delta)A(J^+(\Delta)\cap (\Sigma+ t { a}))=A(J^+(\Delta)\cap (\Sigma+ t {a})) A(\Delta)= A(\Delta).$$ On the other hand, if the distance in $\Sigma$ between $\Delta$ and $\Delta'$ is strictly positive, then  $$(\Delta'+ t{a})\cap J^+(\Delta) = \emptyset \quad\mbox{for all sufficiently small $t>0$,}$$  so that $A(J^+(\Delta)\cap (\Sigma+ t {a}))A(\Delta'+ t{a})=0$. Finally,  
$$A(\Delta)A(\Delta'+ t{ a}) = (A(J^+(\Delta)\cap (\Sigma+ t { a}))A(\Delta))A(\Delta'+ t{ a})  $$ $$=A(\Delta) A(J^+(\Delta)\cap (\Sigma+ t {a}))A(\Delta'+ t{a})=0\:.$$ 
In particular, $[A(\Delta),A(\Delta'+ t{a})]= A(\Delta)A(\Delta'+ t{a})- A(\Delta'+ t{a})A(\Delta) $ $$= A(\Delta)A(\Delta'+ t{a})- (A(\Delta)A(\Delta'+ t{a}))^\dagger=0.$$ The case $t<0$ is analogous. 

\section{Proof of Theorem \ref{GMT}}
 Set $A:=||T||^{-1}T$, $p:=tr(\rho A)$, and $S:=\sqrt{A}\rho\sqrt{A}$. Thus, $p\geq 1-\delta$. Lemma I.4 in \cite{W99} gives Winter's original estimate
$$||\rho-S||_1\leq \sqrt{8(1-p)}\leq \sqrt{8\delta},$$
which is the estimate obtained directly from Winter's lemma. The constant in our statement follows by refining the same argument using Hilbert--Schmidt operators. Since $\sqrt{\rho}$ is Hilbert--Schmidt and $||\sqrt{\rho}||_2=1$, we have
$$\rho-S=(I-\sqrt{A})\rho+\sqrt{A}\rho(I-\sqrt{A}).$$
H\"older's inequality for Schatten classes, $||XY||_1\leq ||X||_2||Y||_2$, therefore yields
$$||\rho-S||_1\leq ||(I-\sqrt{A})\sqrt{\rho}||_2||\sqrt{\rho}||_2 +||\sqrt{A}\sqrt{\rho}||_2||\sqrt{\rho}(I-\sqrt{A})||_2.$$
Now
$||\sqrt{A}\sqrt{\rho}||_2^2=tr(\rho A)=p\leq 1$
and
$||(I-\sqrt{A})\sqrt{\rho}||_2^2=||\sqrt{\rho}(I-\sqrt{A})||_2^2=tr\big(\rho(I-\sqrt{A})^2\big).$
Since $(I-\sqrt{A})^2\leq I-A$, it follows that
$tr\big(\rho(I-\sqrt{A})^2\big)\leq tr\big(\rho(I-A)\big)=1-p.$
Consequently,
$||\rho-S||_1\leq 2\sqrt{1-p}.$
Moreover, since $S\geq 0$ and $tr(S)=p$,
$$\left\|S-\frac{S}{p}\right\|_1=1-p.$$
Hence, by the triangle inequality,
$$\left\|\rho-\frac{S}{p}\right\|_1\leq 2\sqrt{1-p}+1-p\leq 2\sqrt{\delta}+\delta.$$
Finally, $S/p=\sqrt{T}\rho\sqrt{T}/tr(T\rho)$, which proves the claim.  $\hfill \Box$

    \end{document}